\documentclass[12pt, draftclsnofoot, onecolumn]{IEEEtran}
\IEEEoverridecommandlockouts
\usepackage{cite}
\usepackage{amsmath,amssymb,amsfonts}
\usepackage{graphicx}
\usepackage{textcomp}
\usepackage{xcolor}
\usepackage[latin1]{inputenc}  
\usepackage[T1]{fontenc}
\usepackage{array}
\usepackage{url}

\def\BibTeX{{\rm B\kern-.05em{\sc i\kern-.025em b}\kern-.08em
    T\kern-.1667em\lower.7ex\hbox{E}\kern-.125emX}}

\usepackage{indentfirst}
\usepackage{booktabs} 
\usepackage{indentfirst}
\usepackage{bm}
\usepackage{subfigure}
\usepackage{balance}
\usepackage{enumitem}
\newtheorem{theorem}{\textbf{Theorem}}

\newtheorem{definition}{\textbf{Definition}}
\newenvironment{proof}{{\noindent\it\textbf{Proof.}}}{\hfill $\square$}
\usepackage{changepage}
\usepackage{algorithm}
\usepackage{algpseudocode}

\usepackage{mathrsfs}
\usepackage[justification=centering]{caption}
\allowdisplaybreaks[4]

\ifodd 1
\newcommand{\com}[1]{\textbf{\color{red} (COMMENT: #1)}} 
\newcommand{\comg}[1]{\textbf{\color{green} (COMMENT: #1)}}
\newcommand{\response}[1]{\textbf{\color{magenta} (RESPONSE: #1)}} 
\else

\newcommand{\com}[1]{}
\newcommand{\comg}[1]{}
\newcommand{\response}[1]{}
\fi
\ifCLASSINFOpdf
\else
\fi
\begin{document}

\title{\LARGE HFEL: Joint Edge Association and Resource Allocation for Cost-Efficient Hierarchical Federated Edge Learning}

\author{\IEEEauthorblockN{Siqi~Luo,
	Xu~Chen,
	Qiong~Wu,
	Zhi~Zhou,
	and
	Shuai Yu}\\
	\textit{School of Data and Computer Science, Sun Yat-sen University, Guangzhou, China}
}

\maketitle

\begin{abstract}
	Federated Learning (FL) has been proposed as an appealing approach to handle data privacy issue of mobile devices compared to conventional machine learning at the remote cloud with raw user data uploading. By leveraging edge servers as intermediaries to perform partial model aggregation in proximity and relieve core network transmission overhead, it enables great potentials in low-latency and energy-efficient FL. Hence we introduce a novel Hierarchical Federated Edge Learning (HFEL) framework in which model aggregation is partially migrated to edge servers from the cloud. We further formulate a joint computation and communication resource allocation and edge association problem for device users under HFEL framework to achieve global cost minimization. To solve the problem, we propose an efficient resource scheduling algorithm in the HFEL framework. It can be decomposed into two subproblems: \emph{resource allocation} given a scheduled set of devices for each edge server and \emph{edge association} of device users across all the edge servers. With the optimal policy of the convex resource allocation subproblem for a set of devices under a single edge server, an efficient edge association strategy can be achieved through iterative global cost reduction adjustment process, which is shown to converge to a stable system point. Extensive performance evaluations demonstrate that our HFEL framework outperforms the proposed benchmarks in global cost saving and achieves better training performance compared to conventional federated learning.
\end{abstract}


%
\IEEEpeerreviewmaketitle

\section{Introduction}
\label{introduction}
As mobile and internet of things (IoT) devices have emerged in large numbers and are generating a massive amount of data \cite{Zhou2019Edge}, \emph{Machine Learning} (ML) has been witnessed to go through a high-speed development due to big data and improving computing capacity, which prompted the development of \emph{Artificial Intelligence} (AI) to revolutionize our life \cite{Goodfellow-et-al-2016}. The conventional ML framework focuses on central data processing, which requires widely distributed mobile devices to upload their local data to a remote cloud for global model training \cite{LI201776}. However, the cloud server is hard to exploit such a multitude of data from massive user devices as it easily suffers external attack and data leakage risk. Given the above threats to data privacy, many device users are reluctant to upload their private raw data to the cloud server \cite{book2019Bart,Privacy2014Gaff,Consumer2013}.

To tackle the data security issue in centralized training, a decentralized ML named \emph{Federated Learning} (FL) is widely envisioned as an appealing approach \cite{konen2016federated}. It enables mobile devices collaboratively build a shared model while preserving privacy sensitive data locally from external direct access. In the prevalent FL algorithm such as \emph{Federated Averaging} (FedAvg), each mobile device trains a model locally with its own dataset and then transmits the model parameters to the cloud for a global aggregation \cite{McMahan2016McMahan}. With the great potential to facilitate large-scale data collection, FL realizes model training in a distributive fashion.

Unfortunately, FL suffers from a bottleneck of communication and energy overhead before reaching a satisfactory model accuracy due to long transmission latency in wide area network (WAN) \cite{article2004Maxim}. With devices' limited computing and communication capacities, plethora of model transmission rounds occur which degrades learning performance under training time budget. And plenty of energy overhead is required for numerous computation and communication iterations which is challenging to low battery devices. In addition, as many ML models are of large size, directly communicating with the cloud over WAN by a massive number of device users could worsen the congestion in backbone network, leading to significant WAN communication latency.

To mitigate such issues, we leverage the power of \emph{Mobile Edge Computing} (MEC), which is regarded as a promising distributed computing paradigm in 5G era for supporting many emerging intelligent applications such as video streaming, smart city and augmented reality \cite{Edge2016Shi}. MEC allows delay-sensitive and computation-intensive tasks to be offloaded from distributed mobile devices to edge servers in proximity, which offers real-time response and high energy efficiency \cite{chen2015efficient,li2019edge,EnergyEfficient2017You}. Along this line, we propose a novel \emph{Hierarchical Federated Edge Learning} (HFEL) framework, in which edge servers usually fixedly deployed with base stations as intermediaries between mobile devices and the cloud, can perform edge aggregations of local models which are transmitted from devices in proximity. When each of them achieves a given learning accuracy, updated models at the edge are transmitted to the cloud for global aggregation. Intuitively, HFEL can help to reduce significant communication overhead over the WAN transmissions between device users and the cloud via edge model aggregations. Moreover, through the coordination by the edge servers in proximity, more efficient communication and computation resource allocation among device users can be achieved. It can enable effective training time and energy overhead reduction.    

Nevertheless, to realize the great benefits of HFEL, we still face the following challenges: \emph{1) how to solve a joint computation and communication resource allocation for each device to achieve training acceleration and energy saving?} The training time to converge to a predefined accuracy level is one of the most important performance metrics of FL. While energy minimization of battery-constrained devices is the main concern in MEC \cite{EnergyEfficient2017You}. Both training time and energy minimization depend on mobile devices' computation capacities and communication resource allocation from edge servers. As the resources of an edge server and its associated devices are generally limited, such optimization is non-trivial to achieve. \emph{2) How to associate a proper set of device users to an edge server for efficient edge model aggregation?} As in Fig. \ref{FL_scenario}, densely distributed mobile devices are generally able to communicate with multiple edge servers. From the perspective of an edge server, it is better to communicate with as many mobile devices as possible for edge model aggregation to improve learning accuracy. While more devices choose to communicate with the same edge server, the less communication resource that each device would get, which brings about longer communication delay. As a result, computation and communication resource allocation for the devices and their edge association issues should be carefully addressed to accomplish cost-efficient learning performance in HFEL.

As a thrust for the grand challenges above, in this paper we formulate a joint computation and communication resource allocation and edge server association problem for global learning cost minimization in HFEL. Unfortunately, such optimization problem is hard to solve. Hence we decompose the original optimization problem into two subproblems: 1) resource allocation problem and 2) edge association problem, and accordingly put forward an efficient integrated scheduling algorithm for HFEL. For resource allocation, given a set of devices which are scheduled to upload local models to the same edge server, we can solve an optimal policy, i.e., the amount of contributed computation capacity of each device and bandwidth resource that each device is allocated to from the edge server. Moreover, for edge association, we can work out a feasible set of devices (i.e., a training group) for each edge server through cost reducing iterations based on the optimal policy of resource allocation within the training group. The iterations of edge association process finally converge to a stable system point, where each edge server owns a stable set of model training devices to achieve global cost efficiency and no edge server will change its training group formation.
\begin{figure}[tp]
	\begin{center}
		\includegraphics[width=\textwidth]{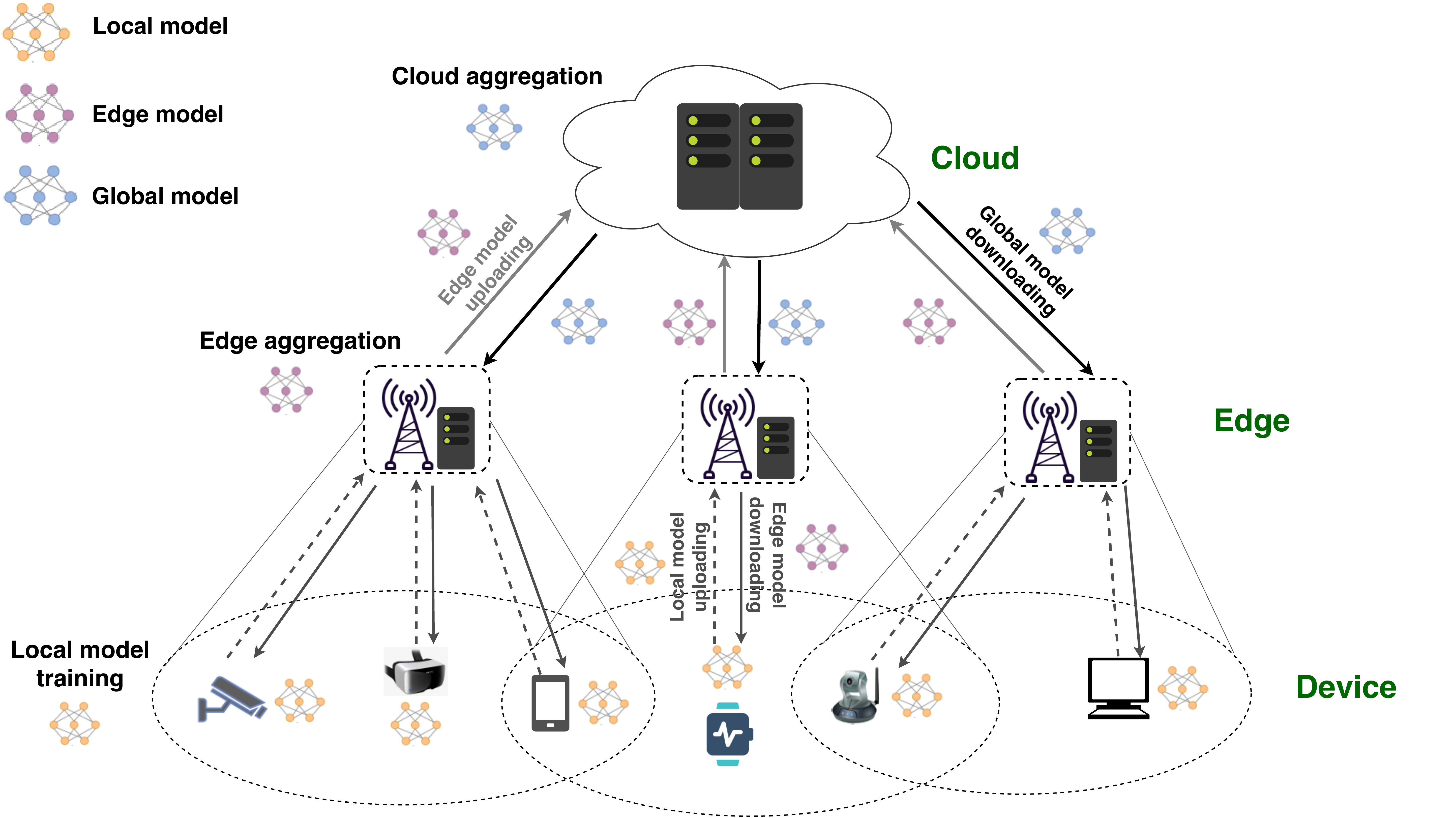}
		\caption{Hierarchical Federated Edge Learning (HFEL) framework.}\label{FL_scenario}
	\end{center}
\end{figure}

In a nutshell, our work makes the key contributions as follows:
\begin{itemize}
	\item{}
	We propose a hierarchical federated edge learning (HFEL) framework which enables great potentials in low latency and energy-efficient federated learning and formulate a holistic joint computation and communication resource allocation and edge association model for global learning cost minimization.
	\item{}
	We decompose the challenging global cost minimization problem into two subproblems: resource allocation and edge association, and accordingly devise an efficient HFEL resource scheduling algorithm. With the optimal policy of the convex resource allocation subproblem given a training group of a single edge server, a feasible edge association strategy can be solved for each edge server through cost reducing iterations which are guaranteed to converge to a stable system point.
	\item{}
	Extensive numerical experiments demonstrate that our HFEL resource scheduling algorithm is capable of achieving superior performance gain in global cost saving over comparing benchmarks and better training performance than conventional device-cloud based FL. 
\end{itemize}

\section{System Model}
\label{model}
{\centering
	\begin{table*}[tp]
		\caption{Key notations.}\label{notation}
		\scriptsize
		\centering
		\begin{tabular}{|c|m{6cm}<{\centering}|c|m{5cm}<{\centering}|} 
			\hline
			\textbf{Symbol}&\textbf{Definitions}&\textbf{Symbol}&\textbf{Definitions}\\
			\hline
			$\mathcal{N}$&set of mobile devices&$\mathcal{K}$&set of edge servers\\
			\hline
			$\mathcal{N}_i$& set of available mobile devices for edge server $i$&$D_n$&device $n$'s training data set\\
			\hline
			$\boldsymbol {x}_j$&the $j$-th input sample of a device&$y_j$&a labeled output of $\boldsymbol {x}_j$ of a device\\
			\hline
			$\theta$&local training accuracy&$\mu$&a constant related to the number of local training iterations\\
			\hline
			$L(\theta)$&number of local iterations&$t$&index of local training iteration\\
			\hline
			$\boldsymbol \omega_n^t$&training model of device $n$ at $t$-th iteration&$\eta_t$&learning rate\\
			\hline
			$c_n$&number of CPU cycles for device $n$ to process one sample data&$f_n^{min},f_n^{max}$&the minimum and maximum computation capacity of device $n$\\
			\hline
			$f_n$&CPU frequency variable of device $n$ for local training&$t_n^{cmp},e_n^{cmp}$&computation delay and energy respectively of $L(\theta)$ local iterations of device $n$\\
			\hline
			$\alpha_n$&effective capacitance coefficient of device $n$'s computing chipset&$\mathcal{S}_i$&set of devices who choose to transmit their model parameters and gradients to edge server $i$\\
			\hline
			$B_i$&edge server $i$'s total bandwidth&$\beta_{i:n}\in(0,1]$&ratio of bandwidth allocated to device $n$ from edge server $i$\\
			\hline
			$r_n$&achievable transmission rate of device $n$&$N_0$&background noise\\
			\hline
			$p_n$&transmission power of device $n$&$h_n$&channel gain of device $n$\\
			\hline
			$t_{i:n}^{com},e_{i:n}^{com}$&communication time and energy respectively for device $n$ to transmit local model to edge server $i$&$d_x(x=n,i)$&device $n$'s or edge server $i$'s update size of model parameters and gradients\\
			\hline
			$\boldsymbol \omega_i$&aggregated model by edge server $i$&$\varepsilon$&edge training accuracy\\
			\hline
			$I(\varepsilon,\theta)$&edge iteration number&$E_{\mathcal{S}_i}^{edge},T_{\mathcal{S}_i}^{edge}$&energy and delay respectively under edge server $i$ with the set of devices $\mathcal{S}_i$\\
			\hline
			$T_i^{cloud},E_i^{cloud}$&delay and energy respectively for edge model uploading by edge server $i$ to the cloud&$r_i$&edge server $i$'s transmission rate to the cloud\\
			\hline
			$p_i$&transmission power of edge server $i$ per second&$D_{\mathcal{S}_i}$&dataset under edge server $i$ with set of devices $\mathcal{S}_i$\\
			\hline
			$D$&total dataset of the set of devices $\mathcal{N}$&$\boldsymbol \omega$&global model aggregated by the cloud under one global iteration\\ 
			\hline
			$E,T$&system-wide energy and delay respectively under one global iteration&$\lambda_e,\lambda_t$&weighting parameters of energy and delay for device training requirements, respectively\\
			\hline
		\end{tabular}
\end{table*}}

In the HFEL framework, we assume a set of mobile devices $\mathcal{N}=\{n:n=1,...,N\}$, a set of edge servers $\mathcal{K}=\{i:i=1,...,K\}$ and a cloud server $S$. Let $\mathcal{N}_i\subseteq\mathcal{N}$ represent the set of available mobile devices communicated with edge server $i$. In addition, each device $n$ owns a local data set $D_n=\{(\boldsymbol{x}_j,y_j)\}_{j=1}^{|D_n|}$ where $\boldsymbol{x}_j$ denotes the $j$-th input sample and $y_j$ is the corresponding labeled output of $\boldsymbol{x}_j$ for $n$'s federated learning task. The key notations used in this paper are summarized in Table \ref{notation}.

\subsection{Learning process in HFEL}
We consider our HFEL architecture as Fig. \ref{FL_scenario}, in which one training model goes through model aggregation in edge layer and cloud layer. Therefore, the shared model parameters by mobile devices in a global iteration involve edge aggregation and cloud aggregation. To quantify training overhead in the HFEL framework, we formulate energy and delay overheads in edge aggregation and cloud aggregation within one global iteration. Note that in most FL scenarios, mobile devices participate in collaborative learning when they are in static conditions such as in the battery-charging state. Hence we assume that in the HFEL architecture, devices remain stable in the learning process, during which their geographical locations keep almost unchanged.

\subsubsection{\textbf{Edge Aggregation}} At this stage, it includes three steps: local model computation, local model transmission and edge model aggregation. That is, local models are first trained by mobile devices and then transmitted to their associated edge servers for edge aggregation, which can be elaborated as the following steps.

\textbf{Step 1. Local model computation.} At this step for a device $n$, it needs to solve the machine learning model parameter $\boldsymbol \omega$ which characterizes each output value $y_j$ with loss function $f_n(\boldsymbol{x}_j,y_j,\boldsymbol \omega)$. The loss function on the data set of device $n$ is defined as
\begin{flalign}
&F_n(\boldsymbol \omega)=\frac {1}{|D_n|}\sum_{j=1}^{|D_n|} f_n(\boldsymbol{x}_j,y_j,\boldsymbol \omega).
\end{flalign}

To achieve a local accuracy $\theta\in(0,1)$ which is common to all the devices for a same model, device $n$ needs to run a number of \emph{local iterations} formulated as $L(\theta)=\mu\log{(1/ \theta)}$ for a wide range of iterative algorithms \cite{Konecny2014Semi}. Constant $\mu$ depends on the data size and the machine learning task. At $t$-th local iteration, each device $n$'s task is to figure out its local update as
\begin{flalign}
&\boldsymbol\omega_n^t=\boldsymbol \omega_n^t-\eta\nabla F_n(\boldsymbol \omega_n^{t-1}), \label{device_local_pro}
\end{flalign}
until $||\nabla F_n(\boldsymbol \omega_n^t)||\leq \theta||\nabla F_n(\boldsymbol \omega_n^{t-1})||$ and $\eta$ is the predefined learning rate \cite{dinh2019federated}. 

Accordingly, the formulation of computation delay and energy overheads incurred by device $n$ can be given in the following. Let $c_n$ be the number of CPU cycles for device $n$ to process one sample data. Considering that each sample $(\boldsymbol{x}_j,y_j)$ has the same size, the total number of CPU cycles to run one local iteration is $c_n|D_n|$. We denote the allocated CPU frequency of device $n$ for computation by $f_n$ with $f_n\in[f_n^{min},f_n^{max}]$. Thus the total delay of $L(\theta)$ local iterations of $n$ can be formulated as
\begin{flalign}
&t_n^{cmp}=L(\theta)\frac{c_n|D_n|}{f_n},
\end{flalign}
and the energy cost of the total $L(\theta)$ local iterations incurred by device $n$ can be given as \cite{Burd1996Processor}
\begin{flalign}
&e_n^{cmp}=L(\theta)\frac {\alpha_n}{2} f_n^2c_n|D_n|,
\end{flalign}
where $\alpha_n/2$ represents the effective capacitance coefficient of device $n$'s computing chipset. 

\textbf{Step 2. Local model transmission.} After finishing $L(\theta)$ local iterations, each device $n$ will transmit its local model parameters $\boldsymbol\omega_n^t$ to a selected edge server $i$, which incurs wireless transmission delay and energy. Then for an edge server $i$, we characterize the set of devices who choose to transmit their model parameters to $i$ as $\mathcal{S}_i\subseteq\mathcal{N}_i$.

In this work, we consider an orthogonal frequency-division multiple access (OFDMA) protocol for devices in which edge server $i$ provides a total bandwidth $B_i$. Define $\beta_{i:n}$ as the bandwidth allocation ratio for device $n$ such that $i$'s resulting allocated bandwidth is $\beta_{i:n}B_i$. Let $r_n$ denote the achievable transmission rate of device $n$ which is defined as
\begin{flalign}
&r_n=\beta_{i:n}B_i\ln{(1+\frac{h_np_n}{N_0})},
\end{flalign}
where $N_0$ is the background noise, $p_n$ is the transmission power, and $h_n$ is the channel gain of device $n$ (which is referred to \cite{vu2019cellfree}). Let $t_{i:n}^{com}$ denote the communication time for device $n$ to transmit $\boldsymbol\omega_n^t$ to edge server $i$ and $d_n$ denote the data size of model parameters $\boldsymbol\omega_n^t$. Thus $t_{i:n}^{com}$ can be characterized by
\begin{flalign}
&t_{i:n}^{com}=d_n/r_n.
\end{flalign}
Given the communication time and transmission power of $n$, the energy cost of $n$ to transmit $d_n$ is
\begin{flalign}
&e_{i:n}^{com}=t_{i:n}^{com}p_n=\frac{d_np_n}{\beta_{i:n}B_i\ln{(1+\frac{h_np_n}{N_0})}}.
\end{flalign}

\textbf{Step 3. Edge model aggregation.} At this step, each edge server $i$ receives the updated model parameters from its connected devices $\mathcal{S}_i$ and then averages them as
\begin{flalign}
&\boldsymbol \omega_i=\frac {\sum_{n\in \mathcal{S}_i} |D_n|\boldsymbol \omega_n^{t}}{|D_{\mathcal{S}_i}|}\label{edge_aggregation},
\end{flalign}
where $D_{\mathcal{S}_i}=\cup_{n\in \mathcal{S}_i}D_n$ is aggregated data set under edge server $i$.

After that, edge server $i$ broadcasts $\boldsymbol \omega_i$ to its devices in $\mathcal{S}_i$ for the next round of local model computation (i.e. step 1). In other words, step 1 to step 3 of edge aggregation will iterate until edge server $i$ reaches an edge accuracy $\varepsilon$ which is the same for all the edge servers. We can observe that each edge server $i$ won't access the local data $D_n$ of each device $n$, thus preserving personal data privacy. In order to achieve the required model accuracy, for a general convex machine learning task, the number of edge iterations is shown to be \cite{Ma2015Distributed}
\begin{flalign}
&I(\varepsilon,\theta)=\frac{\delta(\log{(1/\varepsilon)})}{1-\theta},
\end{flalign}
where $\delta$ is some constant that depends on the learning task. Note that our analysis framework can also be applied when the relation between the convergence iterations and model accuracy is known in non-convex learning tasks.

Since an edge server typically has strong computing capability and stable energy supply, the edge model aggregation time and energy cost for broadcasting the aggregated model parameter $\boldsymbol \omega_i$ is not considered in our optimization model. Since the time and energy cost for a device receiving edge aggregated model parameter $\boldsymbol \omega_i$ is small compared to uploading local model parameters, and keeps almost constant during each iteration, we also ignore this part in our model. Thus, after $I(\varepsilon,\theta)$ edge iterations, the total energy cost of edge server $i$'s training group $\mathcal{S}_i$ is given by
\begin{flalign}
&E_{\mathcal{S}_i}^{edge}=\sum_{n\in \mathcal{S}_i}I(\varepsilon,\theta)(e_{i:n}^{com}+e_n^{cmp}).
\end{flalign}
Similarly, the delay including computation and communication for edge server $i$ to achieve an edge accuracy $\varepsilon$ can be derived as
\begin{flalign}\label{edge_delay_cost}
&T_{\mathcal{S}_i}^{edge}=I(\varepsilon,\theta)\max_{n\in \mathcal{S}_i}\{t_{i:n}^{com}+t_n^{cmp}\}.
\end{flalign}
From (\ref{edge_delay_cost}), we notice that the bottleneck of the computation delay is affected by the last device who finishes all the local iterations, while the communication delay bottleneck is determined by the device who spends the longest time in model transmission after local training.

\subsubsection{\textbf{Cloud Aggregation}} At this stage, we have two steps: edge model uploading and cloud model aggregation. That is, each edge server $i\in\mathcal{K}$ uploads $\boldsymbol \omega_i$ to the cloud for global aggregation after $I(\varepsilon,\theta)$ times edge aggregation.

\textbf{Step 1. Edge model uploading.} Let $r_i$ denote the edge server $i$'s transmission rate to the remote cloud for edge model uploading, $p_i$ the transmission power per sec and $d_i$ the edge server $i$'s model parameter size. We then derive the delay and energy for edge model uploading by edge server $i$ respectively as
\begin{flalign}
&T_i^{cloud}=\frac {d_i}{r_i},\\
&E_i^{cloud}=p_iT_i^{cloud}.
\end{flalign}

\textbf{Step 2. Cloud model aggregation.} At this final step, the remote cloud receives the updated models from all the edge servers and aggregates them as:
\begin{flalign}
&\boldsymbol \omega=\frac {\sum_{i\in\mathcal{K}} |D_{\mathcal{S}_i}|\boldsymbol \omega_i}{|D|},\label{cloud_pro}
\end{flalign}
where $D=\cup_{i\in\mathcal{K}} D_{\mathcal{S}_i}$.

As a result, neglecting the aggregation time at cloud which is much smaller than that on the mobile devices, we can obtain the system-wide energy and delay under one global iteration as
\begin{flalign}
&E=\sum_{i\in\mathcal{K}} (E_i^{cloud}+E_{\mathcal{S}_i}^{edge}),\\
&T=\max_{i\in\mathcal{K}} \{T_i^{cloud}+T_{\mathcal{S}_i}^{edge}\}\label{cloud_delay_cost}.
\end{flalign}
For a more clear description, we provide one global aggregation iteration procedure of HFEL in Algorithm \ref{HFEL_alg}. Such global aggregation procedure can be repeated by pushing the global model parameter $\boldsymbol \omega$ to all the devices via the edge servers, until the stopping condition (e.g., the model accuracy or total training time) is satisfied.

\begin{algorithm}[t]
	\caption{HFEL under one global iteration}\label{HFEL_alg}
	\begin{algorithmic}[1]
		\Require
		Initial models of all the devices $\{\boldsymbol\omega_{n\in\mathcal{N}}^0\}$ with local iteration $t=0$, local accuracy $\theta$,  edge accuracy $\varepsilon$;
		\Ensure
		Global model $\boldsymbol \omega$;
		\\
		\State \textbf{Edge aggregation:}
		\For{$t=1,2,...,I(\varepsilon,\theta)L(\theta)$}
			\For{each device $n=1,...,N$ in parallel}
			\State $n$ solves local problem (\ref{device_local_pro}) and derives $\boldsymbol\omega_n^t$. (\textbf{Local model computation})
			\EndFor
			\State All the devices transmit their updated $\boldsymbol\omega_n^t$ to edge server $i$. (\textbf{Local model transmission})
			\State
			\If{$t~\%~L(\theta)=0$}
				\For{each edge server $i=1,...,K$ in parallel}
				\State $i$ calculates (\ref{edge_aggregation}) after receiving $\{\boldsymbol\omega_n^t:n\in\mathcal{S}_i\}$, and obtains $\boldsymbol \omega_i$. (\textbf{Edge model aggregation})
				\State $i$ broadcasts $\boldsymbol \omega_i$ to $\mathcal{S}_i$ such that $\boldsymbol \omega_n^t=\boldsymbol \omega_i,\forall n \in\mathcal{S}_i$.
				\EndFor
			\EndIf
			
		\EndFor
		\\
		\State \textbf{Cloud aggregation:}
		\State After receiving $\{\boldsymbol \omega_{i\in\mathcal{K}}\}$, the cloud solves problem (\ref{cloud_pro}) and derives the global model $\boldsymbol \omega$.
		
	\end{algorithmic}
\end{algorithm}

\subsection{Problem Formulation}
Given the system model above, we now consider the system-wide performance optimization with respect to energy and delay minimization within one global iteration. Let $\lambda_e,\lambda_t\in[0,1]$ represent the importance weighting indicators of energy and delay for the training objectives, respectively. Then the HEFL optimization problem is formulated as follows:
\begin{flalign}
&min\ \lambda_e E+\lambda_t T,\label{system_minimization}\\
&subject~to,\nonumber\\
&\sum_{n\in \mathcal{S}_i} \beta_{i:n}\leq 1,\forall i\in \mathcal{K},\tag{\ref{system_minimization}a}\\
&0< \beta_{i:n} \le 1,\forall n\in\mathcal{S}_i,\forall i\in\mathcal{K},\tag{\ref{system_minimization}b}\\
&f_n^{min}\leq f_n\leq f_n^{max},\forall n\in\mathcal{N},\tag{\ref{system_minimization}c}\\
&\mathcal{S}_i\subseteq\mathcal{N}_i,\forall i\in \mathcal{K},\tag{\ref{system_minimization}d}\\
&\cup_{i\in \mathcal{K}}\mathcal{S}_i=\mathcal{N},\tag{\ref{system_minimization}e}\\
&\mathcal{S}_i \cap \mathcal{S}_k=\varnothing,\forall i,k\in\mathcal{K}\ and\ i\ne k,\tag{\ref{system_minimization}f}
\end{flalign}
where (\ref{system_minimization}a) and (\ref{system_minimization}c) respectively represent the uplink communication resource constraints and computation capacity constraints, (\ref{system_minimization}d) and (\ref{system_minimization}e) ensure all the devices in the system participate in the model training, and (\ref{system_minimization}f) requires that each device is allowed to associate with one edge server for model parameter uploading and aggregation for sake of cost saving. 

Unfortunately, this optimization problem is hard to solve due to the large combinatorial search space of the edge association decision constraints (\ref{system_minimization}d)-(\ref{system_minimization}f) and their coupling with computation and communication resource allocation in the objective function. This implies that for large inputs it is impractical to obtain the global optimal solution in a real-time manner. Thus, efficient approximating algorithm with low-complexity is highly desirable and this motivates the HFEL scheduling algorithm design in the following.

\subsection{Overview of HFEL Scheduling Scheme}
\begin{figure}[tp]
	\begin{center}
		\includegraphics[width=0.5\textwidth]{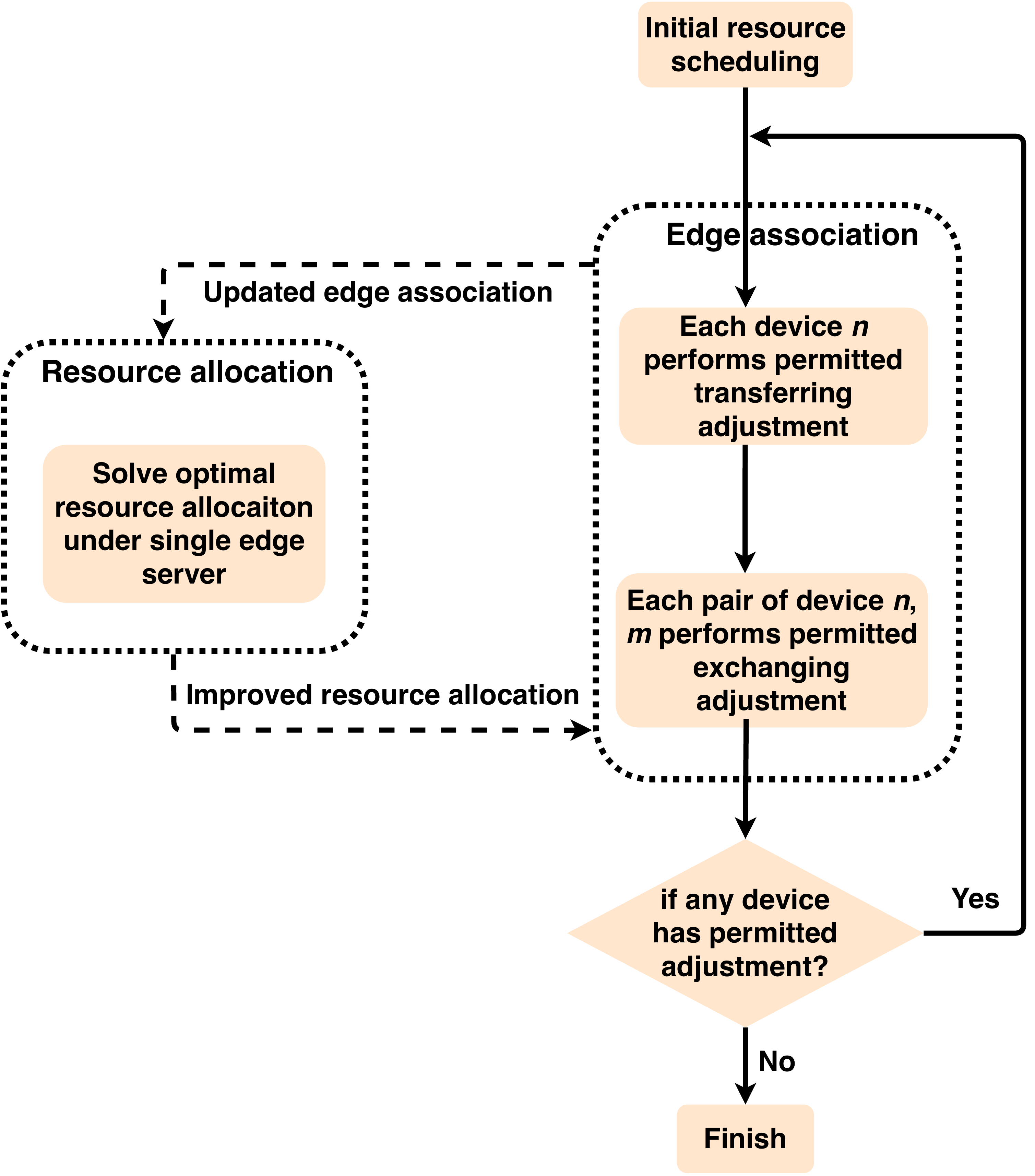}
		\caption{Basic procedure of HFEL scheduling policy.}\label{overview}
	\end{center}
\end{figure}

Since the optimization problem (\ref{system_minimization}) is hard to solve directly, a common and intuitive solution is to design a feasible and computation efficient approach to approximately minimize the system cost. Here we adopt the divide-and-conquer principle and decompose the HFEL scheduling algorithm design issue into two key subproblems: resource allocation within a single edge server and edge association across multiple edge servers.

As shown in Fig. \ref{overview}, the basic procedures of our scheme are elaborated as follows:
\begin{itemize}
	\item{}
	We first carry out an initial edge association strategy (e.g., each device connects to its closest edge server). Given the initial edge association strategy, we then solve the optimal resource allocation for the devices within each edge sever (which is given in Section \ref{single_server} later on).
	\item{}
	Then we define that for each device, it has two possible adjustments to perform to improve edge association scheme: transferring or exchanging (which will be formally defined in Section \ref{multiple_servers} later on). These adjustments are permitted to carry out if they can improve the system-wide performance without damaging any edge server's utility.
	\item{}
	When a device performs a permitted adjustment, it incurs a change of systematic edge association strategy. Thus we will work out the optimal resource allocation for each edge server with updated edge association.
	\item{}
	All the devices iteratively perform possible adjustments until there exists no permitted adjustment, i.e., no change of systematic edge association strategy.
\end{itemize}

As shown in the following sections, the resource allocation subproblem can be efficiently solved in practice using convex optimization solvers, and the edge association process can converge to a stable point within a limited number of iterations. Hence the resource scheduling algorithm for HFEL can converge in a fast manner and is amendable for practical implementation.

\section{Optimal Resource Allocation Within Single Edge Server}
\label{single_server}
In this section, we concentrate on the optimal overhead minimization within a single edge server, i.e., considering joint computation and communication resource allocation subproblem under edge server $i$ given scheduled training group of devices $\mathcal{S}_i$.

To simplify the notations, we first introduce the following terms:
\begin{flalign}
A_n=&\frac{\lambda_eI(\varepsilon,\theta)d_np_n}{B_i\ln{(1+\frac{h_np_n}{N_0})}}, \nonumber\\
B_n=&\lambda_eI(\varepsilon,\theta)L(\theta)\frac {\alpha_n}{2}c_n|D_n|, \nonumber\\
W=&\lambda_tI(\varepsilon,\theta), \nonumber\\
D_n=&\frac{d_n}{B_i\ln{(1+\frac{h_np_n}{N_0})}},\nonumber\\
E_n=&L(\theta)c_n|D_n|,\nonumber
\end{flalign}
where $A_n,B_n,D_n,E_n$ and $W$ are constants related to device $n$'s parameters and system setting. Then through refining and simplifying the aforementioned formulation (\ref{system_minimization}) in a single edge server scenario, we can derive a subproblem formulation of edge server $i$'s overhead minimization under one global iteration as follows:
\begin{flalign}
min\ C_i=&\lambda_eE_{\mathcal{S}_i}^{edge}(f_n,\beta_{i:n})+\lambda_tT_{\mathcal{S}_i}^{edge}(f_n,\beta_{i:n}) \label{single_server_minimization}\\
=&\sum_{n\in\mathcal{S}_i}(\frac{A_n}{\beta_{i:n}}+B_nf^2_n)+W\max_{n\in \mathcal{S}_i}\{\frac{D_n}{\beta_{i:n}}+\frac{E_n}{f_n}\},\nonumber\\
&subject~to,\nonumber\\
&0< \sum_{n\in \mathcal{S}_i} \beta_{i:n}\leq 1,\tag{\ref{single_server_minimization}a}\\
&f_n^{min}\leq f_n\leq f_n^{max},\forall n\in\mathcal{S}_i,\tag{\ref{single_server_minimization}b}\\
&0< \beta_{i:n} \le 1,\forall n\in\mathcal{S}_i.\tag{\ref{single_server_minimization}c}
\end{flalign}

For the optimization problem (\ref{single_server_minimization}), we can show it is a convex optimization problem as stated in the following. 

\begin{theorem}\label{theorem_convexity}
The resource allocation subproblem (\ref{single_server_minimization}) is convex.
\end{theorem}

\begin{proof}
The subformulas of $C_i$ consist of three parts: 1) $\frac{A_n}{\beta_{i:n}}$, 2) $B_nf^2_n$ and 3) $\max_{n\in \mathcal{S}_i}\{\frac{D_n}{\beta_{i:n}}+\frac{E_n}{f_n}\}$, each of which is intuitively convex in its domain and all constraints get affine such that problem (\ref{single_server_minimization}) is convex.
\end{proof}

By exploiting the Karush-Kuhn-Tucker (KKT) conditions of problem (\ref{single_server_minimization}), we can obtain the following structural result.

\begin{theorem}\label{theorem_BW_optimal}
	The optimal solutions to device $n$'s bandwidth and computation capacity allocations $\beta_{i:n}^*$ and $f_n^*$ under edge server $i$ of (\ref{single_server_minimization}) satisfy
	\begin{flalign}
		&\beta_{i:n}^*=\frac{(A_n+\frac{2B_nf_n^{*3}}{E_n}D_n)^{\frac{1}{3}}}{\sum_{n\in \mathcal{S}_i}(A_n+\frac{2B_nf_n^{*3}}{E_n}D_n)^{\frac{1}{3}}}.\label{BW_optimal_expr}
	\end{flalign}
\end{theorem}

\begin{proof}
First to make (\ref{single_server_minimization}) better tractable, let $t=\max_{n\in \mathcal{S}_i}\{\frac{D_n}{\beta_{i:n}}+\frac{E_n}{f_n}\}$ and $t\geq \frac{D_n}{\beta_{i:n}}+\frac{E_n}{f_n},\forall n\in \mathcal{S}_i$. Then problem ($\ref{single_server_minimization}$) can be further transformed to 
\begin{flalign}
min\ C_i&=\sum_{n\in\mathcal{S}_i}(\frac{A_n}{\beta_{i:n}}+B_nf^2_n)+Wt,\label{further_single_server_minimization}\\
&subject~to,\nonumber\\
&\sum_{n\in \mathcal{S}_i} \beta_{i:n}\leq 1,\tag{\ref{further_single_server_minimization}a}\\
&f_n^{min}\leq f_n\leq f_n^{max},\forall n\in\mathcal{S}_i,\tag{\ref{further_single_server_minimization}b}\\
&0< \beta_{i:n} \le 1,\forall n\in\mathcal{S}_i,\tag{\ref{further_single_server_minimization}c}\\
&\frac{D_n}{\beta_{i:n}}+\frac{E_n}{f_n}\leq t,\forall n\in \mathcal{S}_i.\tag{\ref{further_single_server_minimization}d}
\end{flalign}
\begin{algorithm}[t]
	\caption{Resource Allocation Algorithm}\label{resource_allocation}
	\begin{algorithmic}[1]
		\Require
		Initial $\{f_n:n\in\mathcal{S}_i\}$ by random setting;
		\Ensure
		Optimal resource allocation policy under edge server $i$ as $\{\beta_{i:n}^*:n\in\mathcal{S}_i\}$ and $\{f_n^*:n\in\mathcal{S}_i\}$.
		
		\State \quad Replace $\beta_{i:n}$ with equation (\ref{BW_optimal_expr}) in problem (\ref{single_server_minimization}), which then is transformed to an equivalent convex optimization problem (\ref{convex_cpu}) with respective to variables $\{f_n:n\in\mathcal{S}_i\}$.
		\State \quad Utilize convex optimization solvers (e.g., CVX and IPOPT) to solve (\ref{convex_cpu}) and obtain optimal computation capacity allocation $\{f_n^*:n\in\mathcal{S}_i\}$.
		\State \quad Given $\{f_n^*:n\in\mathcal{S}_i\}$, optimal bandwidth allocation $\{\beta_{i:n}^*:n\in\mathcal{S}_i\}$ can be derived based on (\ref{BW_optimal_expr}).
	\end{algorithmic}
\end{algorithm}
Given $\mathcal{S}_i,\forall i\in\mathcal{K}$, problem (\ref{further_single_server_minimization}) is convex such that it can be solved by the Lagrange multiplier method. The partial Lagrange formula can be expressed as
\begin{flalign}
L_i&=\sum_{n\in\mathcal{S}_i}(\frac{A_n}{\beta_{i:n}}+B_nf^2_n)+Wt+\phi(\sum_{n\in \mathcal{S}_i} \beta_{i:n}-1)+
\sum_{n\in\mathcal{S}_i}\tau_n(\frac{D_n}{\beta_{i:n}}+\frac{E_n}{f_n}-t),\nonumber
\end{flalign}
where $\phi$ and $\tau_n$ are the Lagrange multipliers related to constraints (\ref{further_single_server_minimization}a) and (\ref{further_single_server_minimization}d). Applying KKT conditions, we can derive the necessary and sufficient conditions in the following.
\begin{flalign}
&\frac{\partial L}{\partial \beta_{i:n}}=\phi\beta_{i:n}-\frac{A_n+\tau_nD_n}{\beta_{i:n}^2}=0,\forall n\in\mathcal{S}_i,\label{KKT_BW}\\
&\frac{\partial L}{\partial f_n}=2B_nf_n-\frac{\tau_nE_n}{f_n^2}=0,\forall n\in\mathcal{S}_i,\label{KKT_cpu}\\
&\frac{\partial L}{\partial t}=W-\sum_{n\in \mathcal{S}_i}\tau_n=0,\forall n\in\mathcal{S}_i,\\
&\phi(\sum_{n\in \mathcal{S}_i} \beta_{i:n}-1)=0,\phi\geq 0,\label{KKT_BW_sum}\\
&\tau_n(\frac{D_n}{\beta_{i:n}}+\frac{E_n}{f_n}-t)=0,\tau_n\geq 0,\forall n\in \mathcal{S}_i,\\
&f_n^{min}\leq f_n\leq f_n^{max},0< \beta_{i:n} \le 1,\forall n\in\mathcal{S}_i.
\end{flalign}
From (\ref{KKT_BW}) and (\ref{KKT_cpu}), we can derive the relations below:
\begin{flalign}
&\phi=\frac{A_n+\tau_nD_n}{\beta_{i:n}^3}>0,\\
&\beta_{i:n}=({\frac{A_n+\tau_nD_n}{\phi}})^{\frac{1}{3}},\\
&\tau_n=\frac{2B_nf_n^3}{En},\label{tau_expr}
\end{flalign}
based on which, another relation expression can be obtained combining (\ref{KKT_BW_sum}) as follows.
\begin{flalign}
&\sum_{n\in\mathcal{S}_i} {(A_n+\tau_nD_n)}^{\frac{1}{3}}=\phi^{\frac{1}{3}}=\frac{(A_n+\tau_nD_n)^{\frac{1}{3}}}{\beta_{i:n}}.
\end{flalign}
Hence, we can easily work out 
\begin{flalign}
&\beta_{i:n}=\frac{(A_n+\tau_nD_n)^{\frac{1}{3}}}{\sum_{n\in\mathcal{S}_i} {(A_n+\tau_nD_n)}^{\frac{1}{3}}}.\label{beta_expr}
\end{flalign}
Finally, replacing $\tau_n$ with (\ref{tau_expr}) in expression (\ref{beta_expr}), the optimal bandwidth ratio $\beta_{i:n}^*$ can be easily figured out as (\ref{BW_optimal_expr}).
\end{proof}

Given the conclusions in Theorem \ref{theorem_convexity} and \ref{theorem_BW_optimal}, we are able to efficiently solve the resource allocation problem (\ref{single_server_minimization}) with Algorithm \ref{resource_allocation}. Likewise, by replacing $\beta_{i:n}$ with (\ref{BW_optimal_expr}), we can transform problem (\ref{single_server_minimization}) to an equivalent convex optimization problem as 
\begin{flalign}
min\ &\sum_{n\in\mathcal{S}_i}(\frac{A_n\sum_{n\in\mathcal{S}_i} {(A_n+\frac{2B_nf_n^{3}}{E_n}D_n)}^{\frac{1}{3}}}{(A_n+\frac{2B_nf_n^{3}}{E_n}D_n)^{\frac{1}{3}}}+B_nf^2_n)+\nonumber\\
&W\max_{n\in \mathcal{S}_i}\{\frac{D_n\sum_{n\in\mathcal{S}_i} {(A_n+\frac{2B_nf_n^{3}}{E_n}D_n)}^{\frac{1}{3}}}{(A_n+\frac{2B_nf_n^{3}}{E_n}D_n)^{\frac{1}{3}}}+\frac{E_n}{f_n}\},\label{convex_cpu}\\
&subject~to,\nonumber\\
&f_n^{min}\leq f_n\leq f_n^{max},\forall n\in\mathcal{S}_i.
\end{flalign} 

Since the original problem (\ref{single_server_minimization}) is convex and $\beta_{i:n}$ is convex with respect to $f_n$, the transformed problem (\ref{convex_cpu}) above is also convex, which can be solved by some convex optimization solvers (e.g., CVX and IPOPT) to obtain optimal solution $\{f_n^*:n\in\mathcal{S}_i\}$. After that, optimal solution $\{\beta_{i:n}^*:n\in\mathcal{S}_i\}$ can be derived based on (\ref{BW_optimal_expr}) given $\{f_n^*:n\in\mathcal{S}_i\}$. Note that by such problem transformation, we can greatly reduce the size of decision variables in the original problem (\ref{single_server_minimization}) which can help to significantly reduce the solution computing time in practice.

\section{Edge Association For multiple edge servers}
\label{multiple_servers}

We then consider the edge association subproblem for multiple edge servers. Given the optimal resource allocation of scheduled devices under a single edge server, the key idea of solving systematic overhead minimization is to efficiently allocate a bunch of devices to each edge server for edge model aggregation. In the following, we will design an efficient edge association for all the edge servers, in order to iteratively improve the overall system performance.

First we introduce some critical concepts and definitions about edge association by each edge server in the following.

\begin{definition}
	In our system, a \textbf{local training group} $\mathcal{S}_i$ is termed as a subset of $\mathcal{N}_i$, in which devices choose to upload their local models to edge server $i$ for edge aggregation. Correspondingly, the utility of $\mathcal{S}_i$ can be derived as $v(\mathcal{S}_i)=- C_i(\boldsymbol f^*,\boldsymbol \beta_i^*)$ which takes a minus sign over the minimum cost of solving resource allocation subproblem for edge server $i$.
\end{definition}

\begin{definition}
	An \textbf{edge association strategy} $DS=\{\mathcal{S}_i:i\in \mathcal{K}\}$ is defined as the set of local training groups of all the edge servers, where $\mathcal{S}_i=\{n:n\in\mathcal{N}_i\}$, such that the system-wide utility given scheduled $DS$ can be denoted as $v(DS)=\sum_{i=1}^{K} v(\mathcal{S}_i)$.
\end{definition}

For the whole system, which kind of edge association strategy it prefers depends on $v(DS)$. To compare different edge association strategies, we define a preference order based on $v(DS)$ which reflects preferences of all the edge servers for different local training group formations.

\begin{definition}
	Given two different edge association strategies $DS^1$ and $DS^2$, we define a \textbf{preference order} as $DS^1\vartriangleright DS^2$ if and only if $v(DS^1)>v(DS^2)$. It indicates that edge association strategy $DS^1$ is preferred over $DS^2$ to gain lower overhead by all the edge servers.
\end{definition}

Next, we can solve the overhead minimization problem by constantly adjusting edge association strategy $DS$, i.e., each edge server's training group formation, to gain lower overhead in accordance with preference order $\vartriangleright$. The edge association adjusting will result in termination with a stable $DS^*$ where no edge server $i$ in the system will deviate its local training group from $\mathcal{S}^*_i\in DS^*$.

Obviously the adjustment of edge association strategy $DS$ basically results from the change of each edge server's local training group formation. In our system, it is permitted to perform some edge association adjustments with utility improvement based on $\vartriangleright$ defined as follows. 

\begin{definition}
	A \textbf{device transferring adjustment} by $n$ means that device $n\in\mathcal{S}_i$ with $|\mathcal{S}_i|>2$ retreats its current training group $\mathcal{S}_i$ and joins another training group $\mathcal{S}_{-i}$. Causing a change from $DS^1$ to $DS^2$, the device transferring adjustment is permitted if and only if $DS^2\vartriangleright DS^1$.
\end{definition}

\begin{definition}
	A \textbf{device exchanging adjustment} between edge servers $i$ and $j$ means that device $n\in\mathcal{S}_i$ and $m\in\mathcal{S}_j$ are switched to each other's local training group. Causing a change from $DS^1$ to $DS^2$, the device exchanging adjustment is permitted if and only if $DS^2\vartriangleright DS^1$.
\end{definition}

Based on the wireless communication between devices and edge servers, each device reports all its detailed information (including computing and communication parameters) to its available edge servers. Then each edge server $i$ will calculate its own utility $v(\mathcal{S}_i)$, communicate with the other edge servers through cellular links and manage the edge association adjustments.

With the iteration of every permitted adjustment which brings a systematic overhead decrease by $\Delta=v(DS^2)-v(DS^1)$, the edge association adjustment process will terminate to be stable where no edge server will deviate from the current edge association strategy.

\begin{definition}
	\label{nash}
	An edge association strategy $DS^*$ is at a \textbf{stable system point} if no edge server $i$ will change $\mathcal{S}_i^*\in DS^*$ to obtain lower global training overhead with $\mathcal{S}_{-i}^*\in DS^*$ unchanged.
\end{definition}

That is, at a stable system point $DS^*$, no edge server $i$ will deviate its local training group formation from $\mathcal{S}_i^*\in DS^*$ to achieve lower global FL overhead given optimal resource allocation within $\mathcal{S}_i^*$.

\begin{algorithm}[t]
	\caption{Edge Association Algorithm}\label{device_schedule}
	\begin{algorithmic}[1]
		\Require
		Set of devices $\mathcal{N}$, tasks $\mathcal{T}$ and sensing data $\mathcal{K}$;
		\Ensure
		Stable system point $DS^*$.
		\For{$i=1$ to $K$}
		\State edge server $i$ randomly forms $\mathcal{S}_i$.
		\State $i$ solves optimal resource allocation and derives an initial $v(\mathcal{S}_i)$ within $\mathcal{S}_i$.
		
		\EndFor
		\State An initial edge association strategy is obtained as $DS$.
		\State
		
		\Repeat
		\For {$n=1$ to $N$}
		\State each pair of edge server $i$ and $j$ with $i\ne j$ perform \textbf{device transferring adjustment} by transferring device $n$ ($n\in\mathcal{S}_i$ and $n\in\mathcal{N}_{j}$) from $\mathcal{S}_i$ to $\mathcal{S}_j$ if permitted. Then $\boldsymbol h_i$ and $\boldsymbol h_j$ are accordingly updated.
		\EndFor
		\State randomly pick device $n\in\mathcal{S}_i$ and $m\in\mathcal{S}_j$ where $i\ne j$, perform \textbf{device exchanging adjustment} if permitted. Then $\boldsymbol h_i$, $\boldsymbol h_j$ and $DS$ are accordingly updated.
		\Until no edge association adjustment is permitted by any device $n\in\mathcal{N}$.
		\State
		\State Obtain the optimal edge association $DS^*$ with each $\mathcal{S}_i\in DS^*$ achieving $f_n^*$ and $\beta_{i:n}^*,n\in\mathcal{S}_i$.
		
	\end{algorithmic}
\end{algorithm}

Next, we devise an edge association algorithm to achieve cost efficiency in HFEL for all the edge servers and seek feasible computation and communication resource allocation for their training groups. Note that in our scenario, each edge server has perfect knowledge of the channel gains and computation capacities of its local training group which can be obtained by feedback. They also can connect with each other through cellular links. Thus, our decentralized edge association process is implemented by all the edge servers, which consists of two steps: \emph{initialized allocation} and \emph{edge association} as described in Algorithm \ref{device_schedule}.

In the first stage, \emph{initialization allocation} procedure is as follows.
\begin{itemize}
	\item{}
	First for each edge server $i\in \mathcal{K}$, local training group $\mathcal{S}_i$ is randomly formed.
	\item{}
	Then given $\mathcal{S}_i$, edge server $i$ solves resource allocation subproblem, i.e., obtaining $f_n^*$ and $\beta_{i:n}^*,\forall n\in\mathcal{S}_i$ and deriving $v(\mathcal{S}_i)$.
	\item{}
	After the initial edge associations of all the edge servers complete, an initial edge association strategy $DS=\{\mathcal{S}_i,...,\mathcal{S}_K\}$ can be achieved.
\end{itemize}

In the second stage, edge servers execute edge association by conducting permitted edge association adjustments in an iterative way until no local training group will be changed. At each iteration, edge servers involved will calculate their own utilities. Specially, a historical group set $\boldsymbol h_i$ is maintained for each edge server $i$ to record the group composition it has formed before with the corresponding utility value so that repeated calculations can be avoided. 

Take device transferring adjustment for example. During an iteration, an edge server $i$ firstly contends to conduct device transferring adjustment. That is, edge server $i$ transfers its device $n$ from $\mathcal{S}_i$ to another edge server $j$'s training group $\mathcal{S}_j$. And we define $\mathcal{S}'_i=\mathcal{S}_i\setminus n$ and $\mathcal{S}'_j=\mathcal{S}_j\cup \{n\}$. This leads to a change of edge association strategy from $DS^1$ to $DS^2$. Secondly, note that each edge server $i$ maintains a historical set $\boldsymbol h_i$ to record the group composition it has formed before with the corresponding utility value. It enables edge server $i$ and $j$ to reckon their utility changes as $\Delta_i=v(\mathcal{S}'_i)-v(\mathcal{S}_i)$ and $\Delta_j=v(\mathcal{S}'_j)-v(\mathcal{S}_j)$, which can also reflect the system-wide utility improvement $\Delta=\Delta_i+\Delta_j=v(DS^2)-v(DS^1)$. Finally, edge server $i$ and $j$ can decide to conduct this device transferring adjustment when $\Delta>0$.

After the edge association algorithm converges, all the involved mobile devices will execute local training with the optimal resource allocation strategy $f_n^*$ and $\beta_{i:n}^*$ that are broadcast from the edge server.

Extensive performance evaluation in Section \ref{experiment} shows that the proposed edge association algorithm can converge in a fast manner, with an almost linear convergence speed. 

\section{Performance Evaluation}\label{experiment}
In this section, we carry out simulations to evaluate: 1) the global cost saving performance of the proposed resource scheduling algorithm and 2) HFEL performance in terms of test accuracy, training accuracy and training loss. From the perspective of devices' and edge servers' availability, all the devices and edge servers are distributed randomly within an entire $500M\times500M$ area.

\begin{table}[h]
	\caption{Simulation settings.}\label{parameter}
	\centering
	\begin{tabular}{|c|c|}
		\hline
		\textbf{Parameter}&\textbf{Value}\\
		\hline
		Maximum Bandwidth of Edge Servers&10 MHz\\
		\hline
		Device Transmission Power&200 mW\\
		\hline
		Device CPU Freq.&[1, 10] GHz\\
		\hline
		Device CPU Power&600 mW\\
		\hline
		Processing Density of Learning Tasks&[30, 100] cycle/bit\\
		\hline
		Background Noise&$10^{-8}$ W\\
		\hline
		Device Training Size&[5, 10] MB\\
		\hline
		Updated Model Size&25000 nats\\
		\hline
		Capacitance Coefficient&$2\times 10^{-28}$\\
		\hline
		Learning rate&0.0001\\
		\hline
	\end{tabular}
\end{table}

\subsection{Performance gain in cost reduction}
Typical parameters of devices and edge servers are provided in Table \ref{parameter} with image classification learning tasks on a dataset MNIST \cite{Lecun1998GradientBased}. To characterize mobile device heterogeneity for MNIST dataset, we have each device maintain only two labels over the total of $10$ labels and their sample sizes are different based on the law power in \cite{li2018federated}. Furthermore, each device trains with full batch size. Under varying device number from $15$ to $60$ and edge server number from $5$ to $25$, we compare our algorithm to the following schemes to present the performance gain in cost reduction with local training accuracy $\theta=0.9$ and edge training accuracy $\varepsilon=0.9$:

\begin{itemize}
	\item{}
	\emph{Random edge association}: each edge server $i$ selects the set of mobile devices $\mathcal{S}_i$ in a random way and then solves the optimal resource allocation for $\mathcal{S}_i$. That is, it only optimizes resource allocation subproblem given a set of devices.
	\item{}
	\emph{Greedy edge association}: each device can select the connected edge server sequentially based on the geographical distance to each edge server in an ascending order. After that, each edge server $i$ solves the optimal resource allocation with $\mathcal{S}_i$. It also only optimizes resource allocation subproblem without edge association similar to random resource allocation.
	\item 
	\emph{Computation optimization}: in this scheme, resource allocation subproblem for each $\mathcal{S}_i,i\in\mathcal{K}$ solves optimal computation capacity $f_{n\in\mathcal{S}_i}^*$ given evenly distribution of bandwidth ratio.
	\item 
	\emph{Communication optimization}: in this scheme, resource allocation subproblem for each $\mathcal{S}_i,i\in\mathcal{K}$ solves optimal bandwidth ratio allocation $\beta_{i:n}^*$ with random computation capacity decision $f_{n\in\mathcal{S}_i}\in[f_n^{min},f_n^{max}]$.
	\item{}
	\emph{Uniform resource allocation}: in this scheme, we leverage the same edge association strategy as our proposed algorithm. While in the resource allocation subproblem, the bandwidth of each edge server $i$ is evenly distributed to mobile devices in $\mathcal{S}_i$ and the computation capacity of $n\in\mathcal{S}_i$ is randomly determined between $f_n^{min}$ and $f_n^{max}$. That is, edge association subproblem is solved without resource allocation optimization. 
	\item{}
	\emph{Proportional resource allocation}: for all the edge servers, we as well adopt edge association strategy to improve $\{\mathcal{S}_i:i\in\mathcal{K}\}$. While in the resource allocation subproblem, the bandwidth of each edge server $i$ is distributed to each $n\in\mathcal{S}_i$ reversely proportional to the distance $l_{i,n}$ such that communication bottle can be mitigated. Similarly, random computation capacity decision of $n$ is $f_{n\in\mathcal{S}_i}\in[f_n^{min},f_n^{max}]$. Similar to uniform resource allocation, only edge association subproblem is solved.
\end{itemize}

\begin{figure}[tp]
	\centering
	\begin{minipage}{0.32\textwidth}
		\centering
		\includegraphics[width=\textwidth]{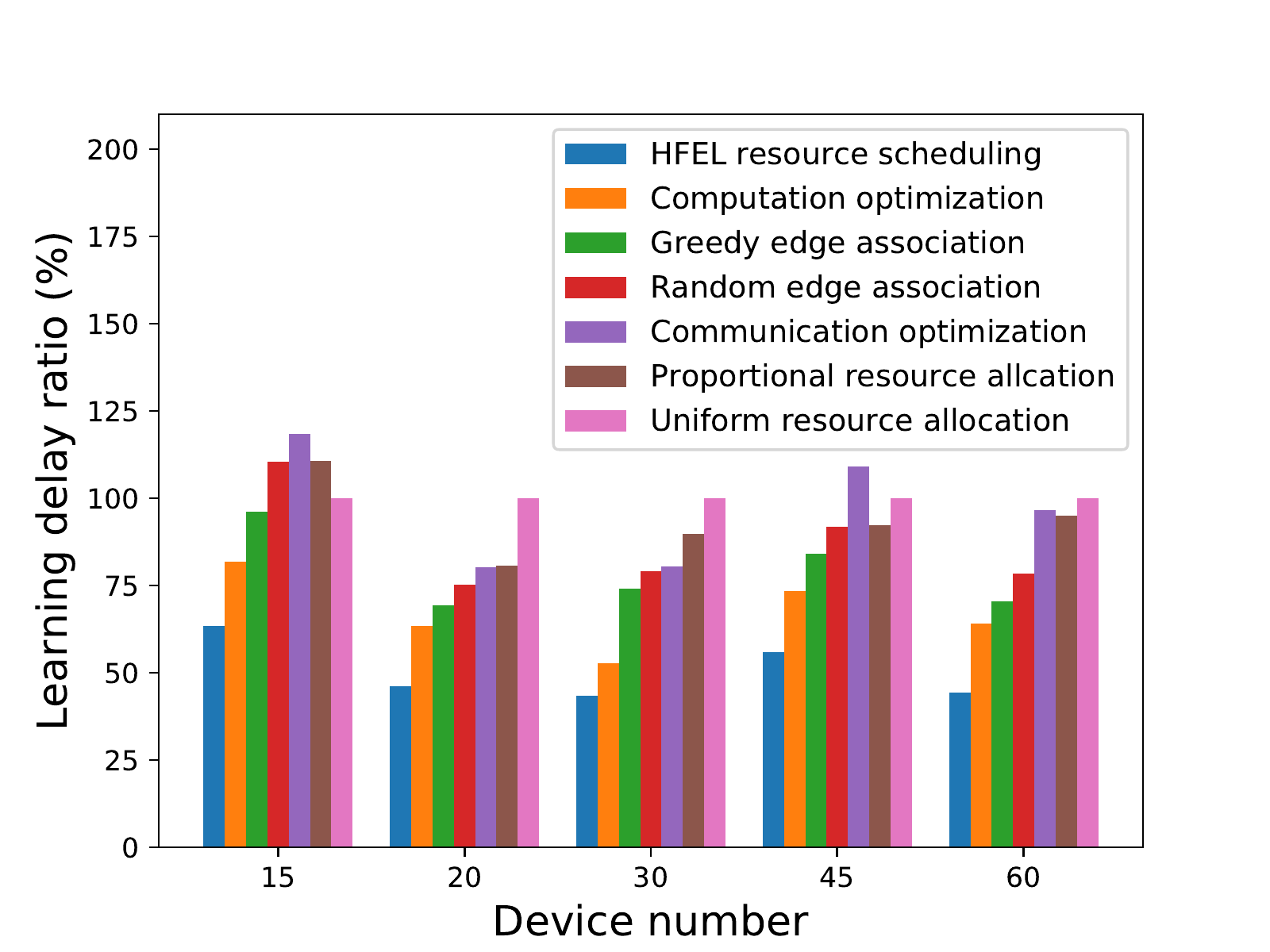}
		\centering
		\caption{Learning delay ratio under growing device number.}\label{delay_under_device}
	\end{minipage}
	\hfill
	\begin{minipage}{0.32\textwidth}
		\centering
		\includegraphics[width=\textwidth]{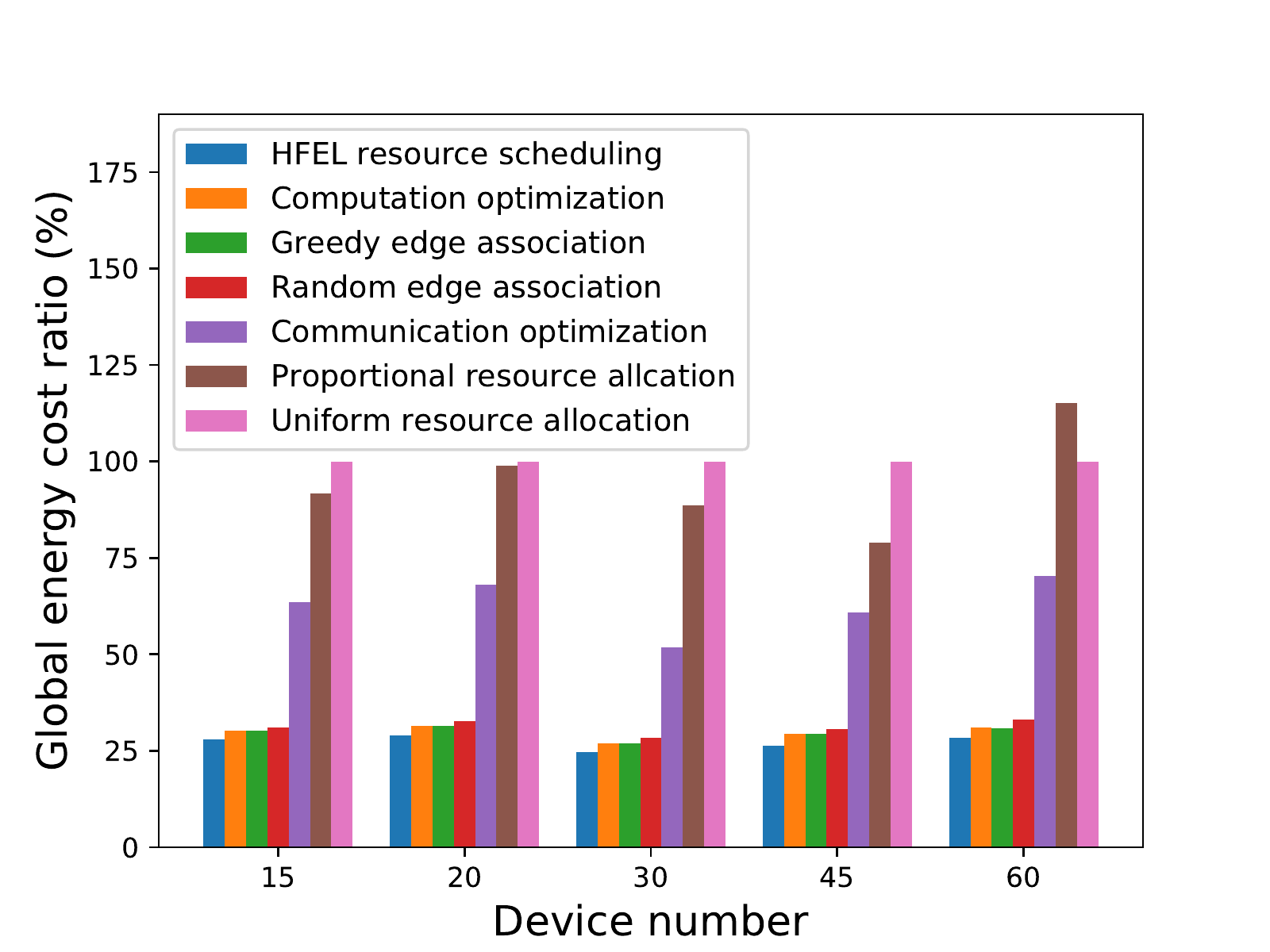}
		\centering
		\caption{Global energy ratio under growing device number.}\label{energy_under_device}
	\end{minipage}
	\hfill
	\begin{minipage}{0.33\textwidth}
		\centering
		\includegraphics[width=\textwidth]{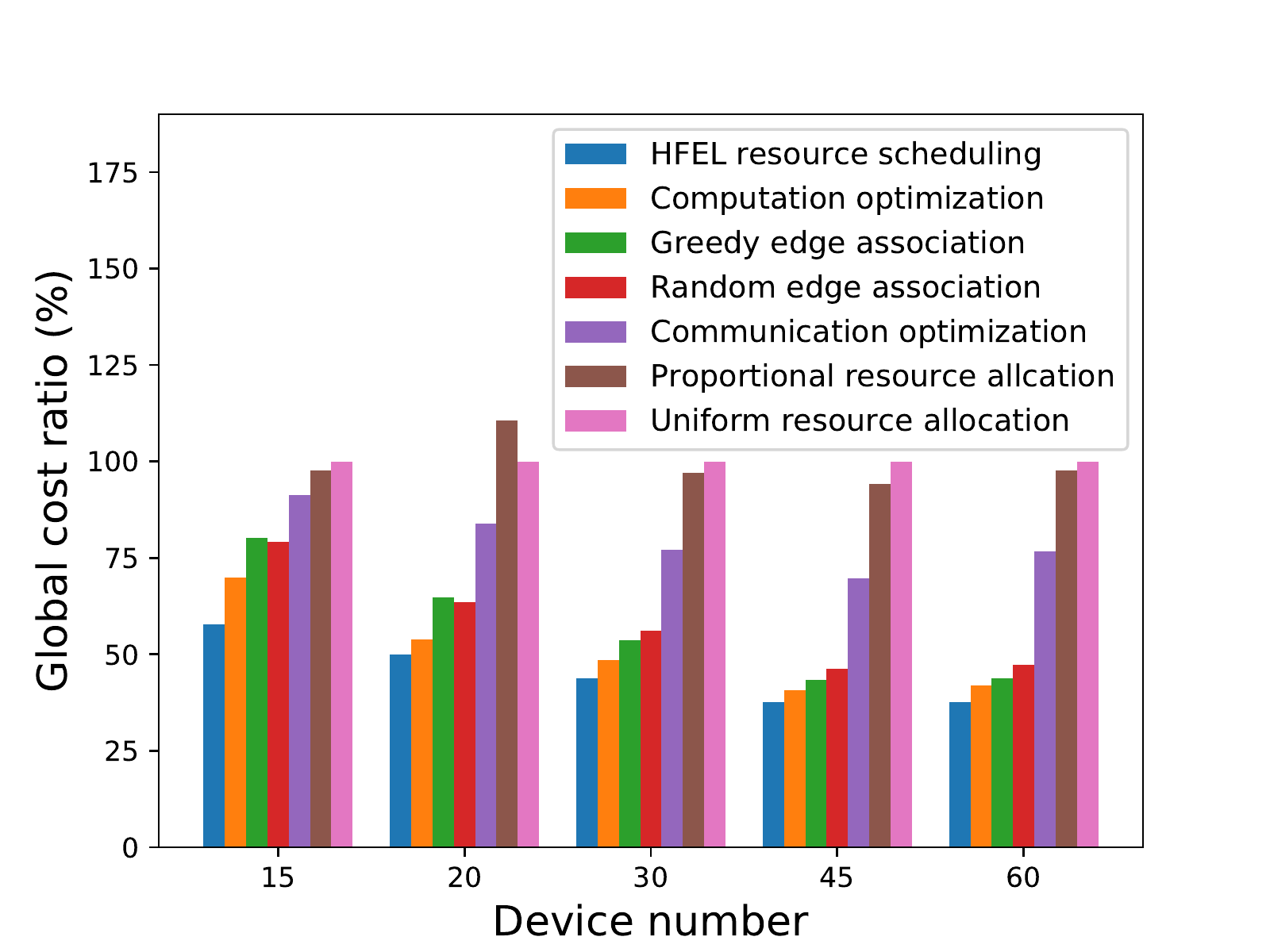}
		\centering
		\caption{Global cost ratio under growing device number.}\label{cost_under_device}
	\end{minipage}
	\vfill
	\begin{minipage}{0.32\textwidth}
		\centering
		\includegraphics[width=\textwidth]{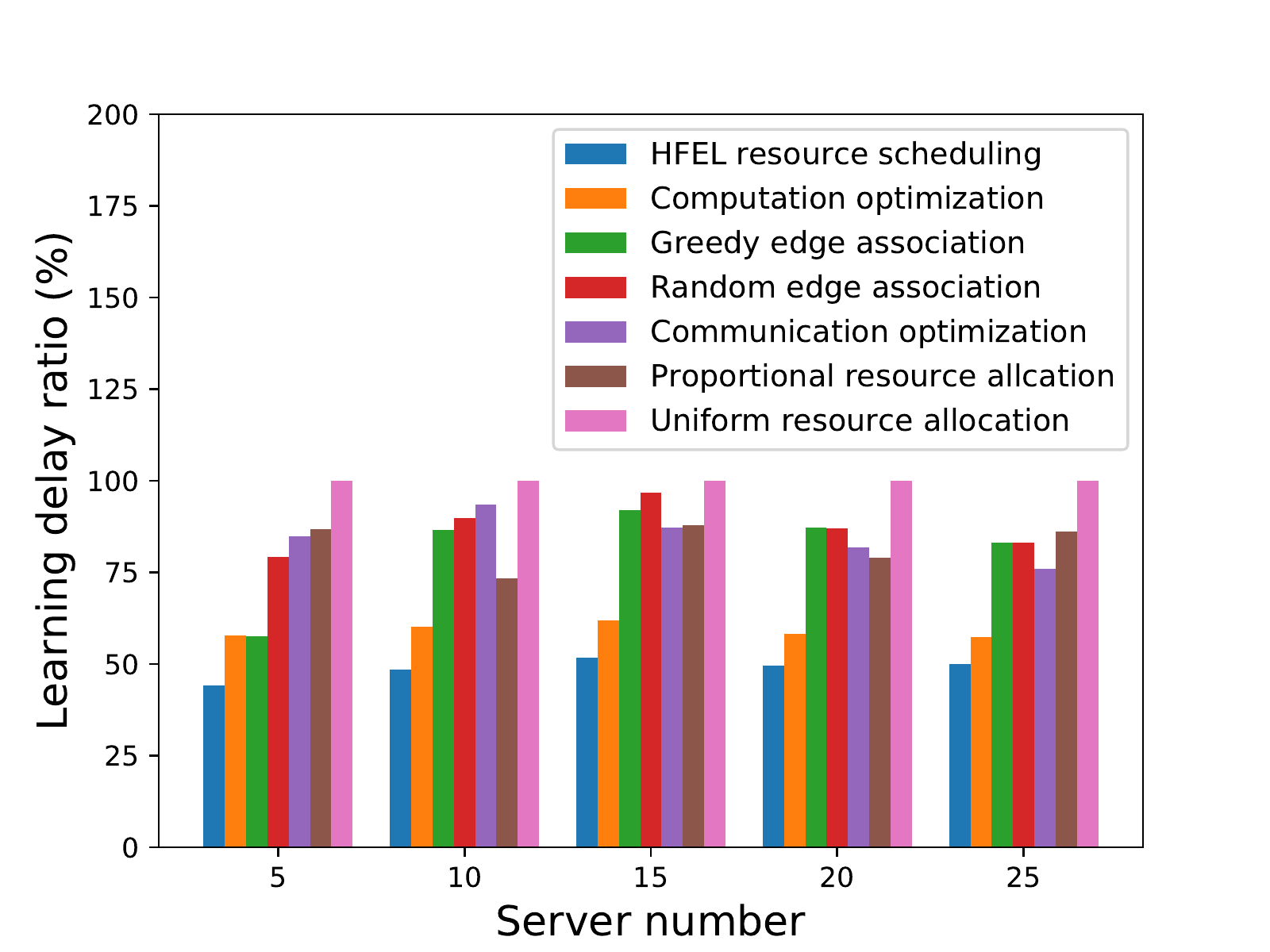}
		\centering
		\caption{Learning delay ratio under growing server number.}\label{delay_under_server}
	\end{minipage}
	\hfill
	\begin{minipage}{0.32\textwidth}
		\centering
		\includegraphics[width=\textwidth]{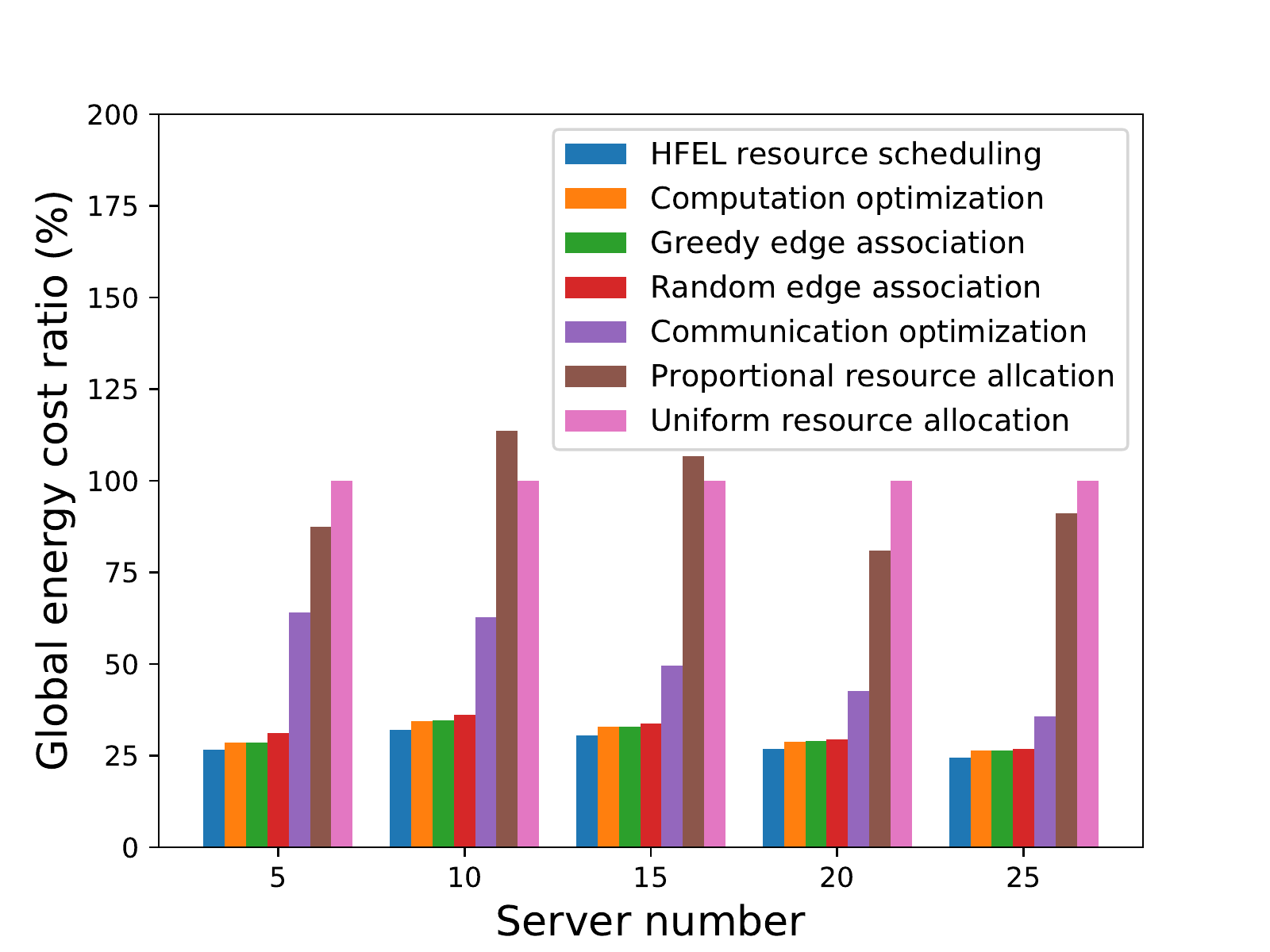}
		\centering
		\caption{Global energy ratio under growing server number.}\label{energy_under_server}
	\end{minipage}
	\hfill
	\begin{minipage}{0.33\textwidth}
		\centering
		\includegraphics[width=\textwidth]{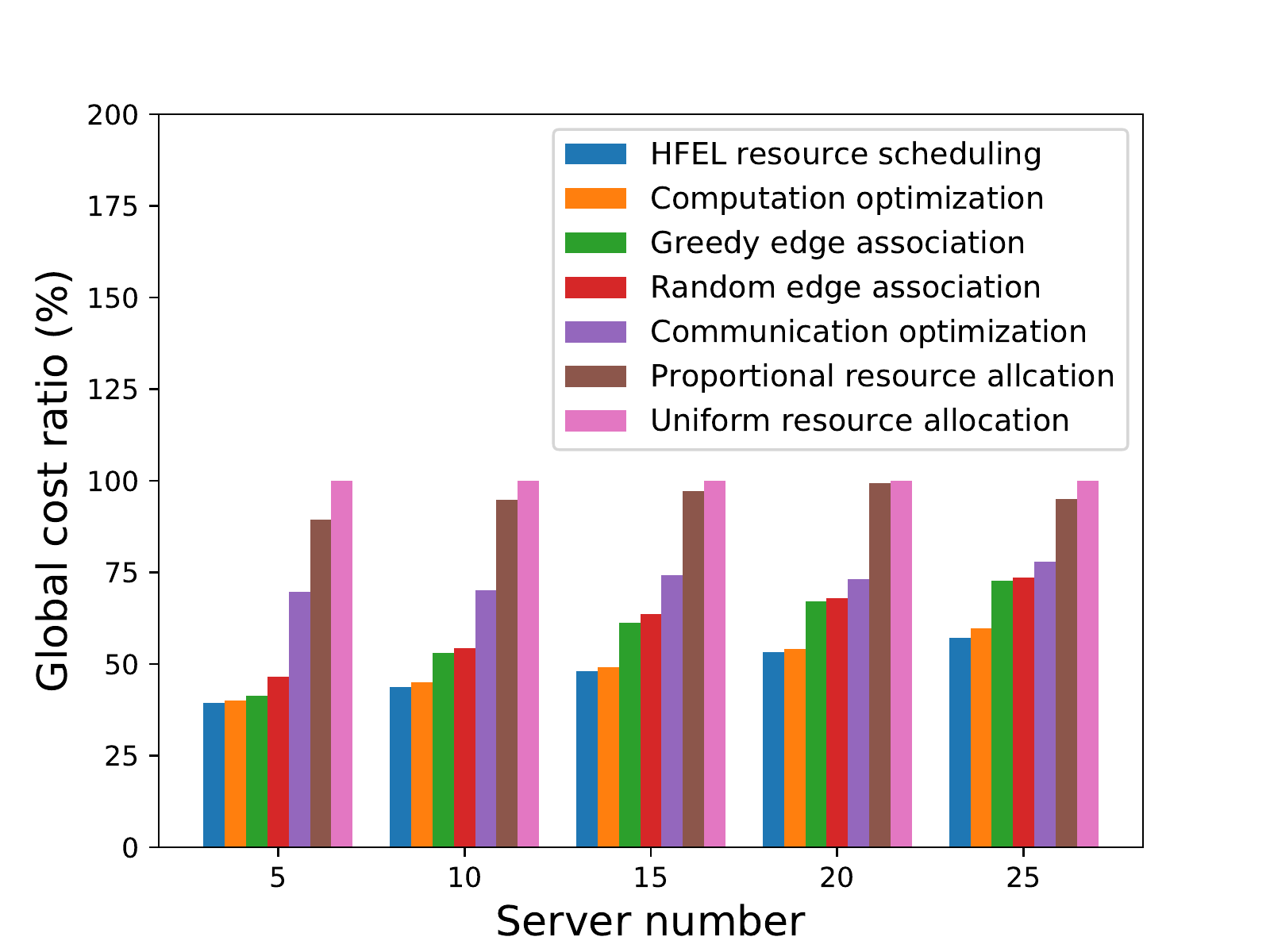}
		\centering
		\caption{Global cost ratio under growing server number.}\label{cost_under_server}
	\end{minipage}
\end{figure}

As presented in Fig. \ref{delay_under_device} to Fig. \ref{cost_under_server} in which uniform resource allocation is regarded as benchmark, our HFEL algorithm achieves the lowest global energy ratio, learning delay ratio and global cost ratio compared to the proposed schemes.

First we explore the impact of different device numbers on the performance gain in cost reduction by fixing edge server number as $5$ in Fig. \ref{delay_under_device} to Fig. \ref{cost_under_device}. Under the weights of energy and delay as $\lambda_e=0$ and $\lambda_t=1$ in Fig. \ref{delay_under_device}, HFEL algorithm accomplishes a satisfying learning delay ratio as $63.3\%, 46.2\%, 43.3\%, 56.0\%$ and $44.4\%$ compared to uniform resource allocation as device number grows. Similarly in Fig. \ref{energy_under_device} with energy and delay weights as $\lambda_e=1$ and $\lambda_t=0$, our HFEL scheme achieves global energy cost ratio as $30\%$ at most compared to uniform resource allocation and $5.0\%$ compared to computation optimization scheme. As described in Fig. \ref{cost_under_device} in which weights of time and energy are randomly assigned, i.e., $\lambda_e,\lambda_t\in [0,1]$ and $\lambda_e+\lambda_t=1$, it shows that HFEL algorithm still outperforms the other six schemes. Compared to computation optimization, greedy device allocation, random device allocation, communication optimization, proportional resource allocation and uniform resource allocation schemes, our algorithm is more efficient and fulfills up to $10\%, 14.0\%, 20.0\%, 51.2\%$, $61.5\%$ and $57.7\%$ performance gain in global cost reduction, respectively.

Then with device number fixed as $60$, Fig. \ref{delay_under_server} to Fig. \ref{cost_under_server} exhibit that our HFEL algorithm still has better performance gain than the other comparing schemes. For example, compared to uniform resource allocation scheme, the HFEL scheme obtains the highest learning delay ratio as $51.6\%$ in Fig. \ref{delay_under_server} and the highest global energy cost ratio as $50.0\%$ in Fig. \ref{energy_under_server}. Meanwhile, Fig. \ref{cost_under_server} presents that our HFEL algorithm can achieve up to $5.0\%, 25.0\%, 24.0\%$, $28.0\%$ and $40.3\%$ global cost reduction ratio over computation optimization, greedy device allocation, random device allocation, communication optimization and proportional resource allocation schemes, respectively.

It is interesting to find that the performance gain of our HFEL scheme compared to the benchmark in global energy ratios as Fig. \ref{energy_under_device} and Fig. \ref{energy_under_server} is better than that in learning delay ratios shown in Fig. \ref{delay_under_device} and Fig. \ref{delay_under_server}. That is because in the objective function, the numerical value of energy cost is much larger than the value of learning delay, which implies that the energy weight plays a leading role in global cost reduction.
\begin{figure}[tp]
	\centering
	\begin{minipage}{0.49\textwidth}
		\centering
		\includegraphics[width=\textwidth]{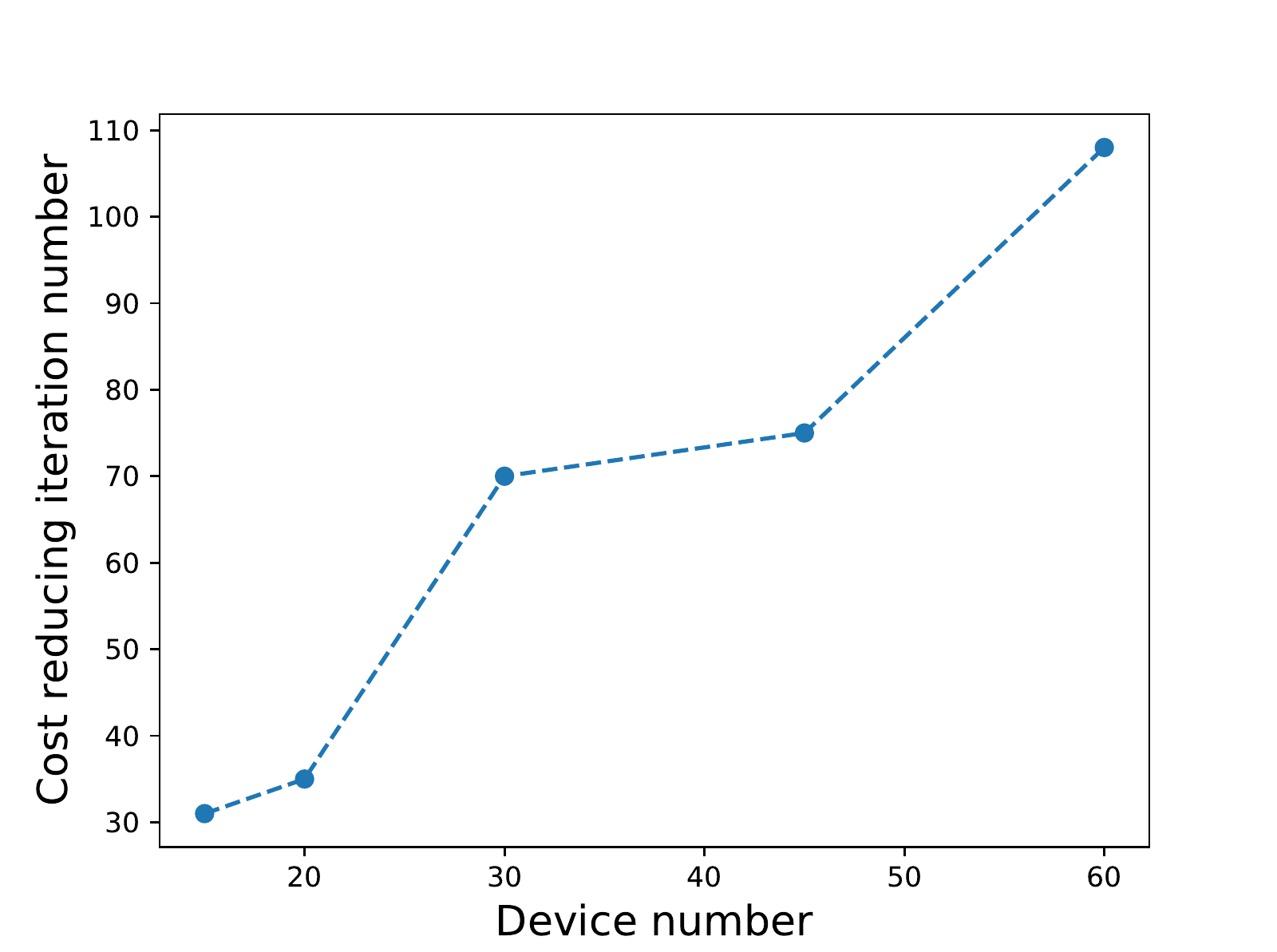}
		\centering
		\caption{Cost reducing iteration number under growing devices.}\label{iteration_num_device}
	\end{minipage}
	\hfill
	\begin{minipage}{0.49\textwidth}
		\centering
		\includegraphics[width=\textwidth]{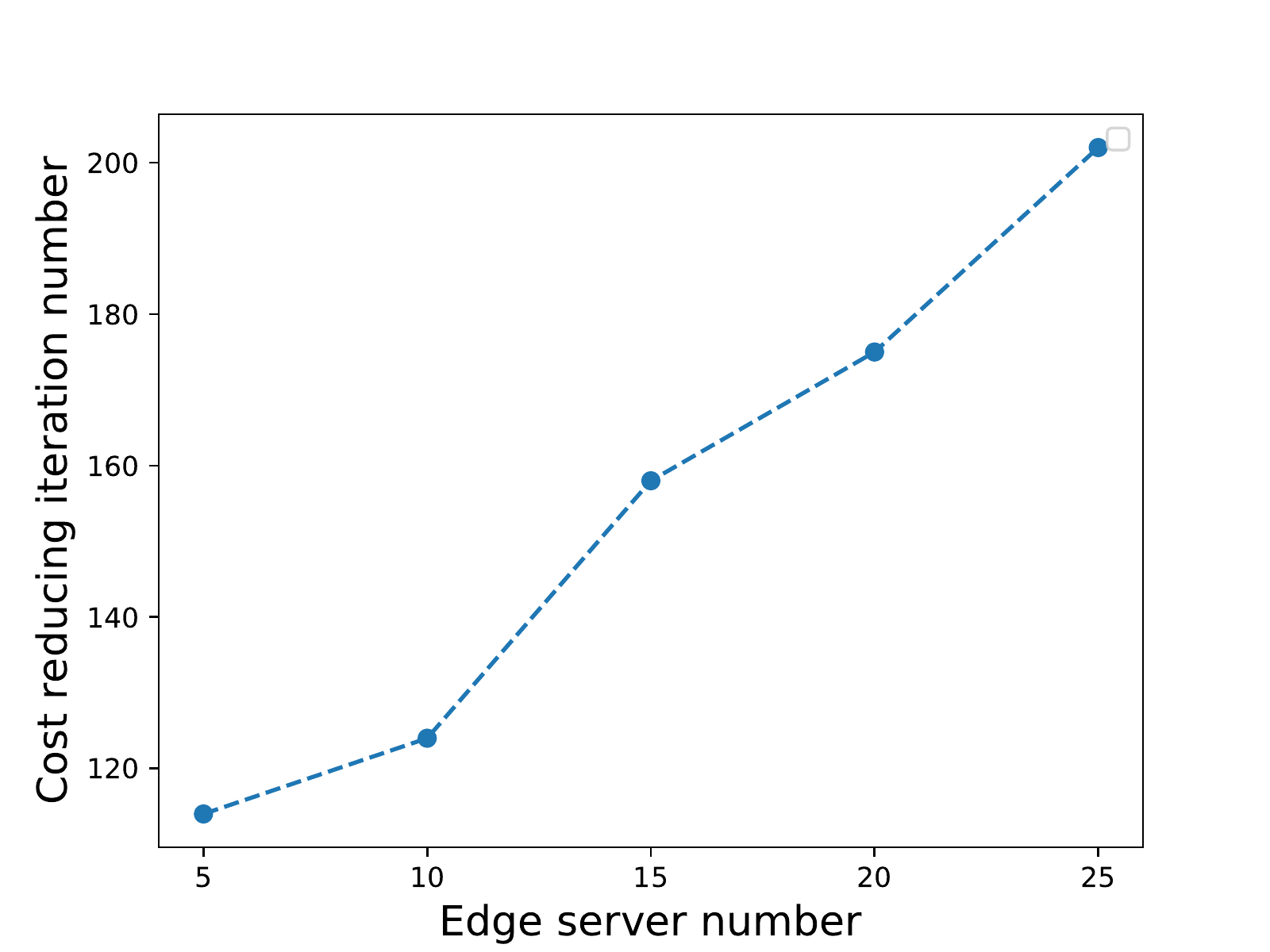}
		\centering
		\caption{Cost reducing iteration number under growing servers.}\label{iteration_num_server}
	\end{minipage}
\end{figure}
Note that greedy device allocation and random device allocation schemes only optimize resource allocation subproblem without edge association. While proportional resource allocation and uniform resource allocation strategies solve edge association without resource allocation optimization. It can be figured out that the performance gain of resource allocation optimization in global cost reduction greatly dominates that of edge association solution.

Further, we show the average iteration number of our algorithm in Fig. \ref{iteration_num_device} with growing number of devices from $15$ to $60$, and the average iteration number of our algorithm in Fig. \ref{iteration_num_server} with the number of edge servers ranging from $5$ to $25$. The results show that the convergence speed of the proposed edge association strategy is fast and grows (almost) linearly as the numbers of mobile device and edge server increase, which reflects the computation efficiency of edge association algorithm.

\subsection{Performance gain in training loss and accuracy}
In this subsection setting, the performance of HFEL is validated on dataset MNIST \cite{Lecun1998GradientBased} and FEMNIST \cite{caldas2018leaf} (an extended MNIST dataset with $62$ labels which is partitioned based on the device of the digit or character) compared to the classic FedAvg algorithm \cite{McMahan2016McMahan}. In addition to different numbers of labels in devices for training on MNIST and FEMNIST dataset, the number of samples of each device varies in different datasets. Specifically, for MNIST and FEMNIST dataset, the number of data samples are in the ranges of [15,4492] and [184,334] in each device \cite{dinh2019federated}, respectively. Moreover, each device trains with full batch size on both MNIST and FEMNIST to perform image classification tasks, which utilize logistic regression with cross-entropy loss function.

We perform training experiments to show the advantages of HFEL scheme over FedAvg, a traditional device-cloud FL architecture not involving edge servers or resource allocation optimization\cite{McMahan2016McMahan}. We consider $5$ edge servers and $30$ devices participating in the training process for experiment. All the datasets are split with $75\%$ for training and $25\%$ for testing in a random way. In the training process, $1000$ global iterations are executed during each of which all the devices go through the same number of local iterations in both HFEL and FedAvg schemes.
\begin{figure*}[t]
	\centering
	\subfigure[Test accuracy.]{
		\label{glob_acc_mnist}
		\begin{minipage}{0.31\textwidth}
			\centering
			\includegraphics[width=\textwidth]{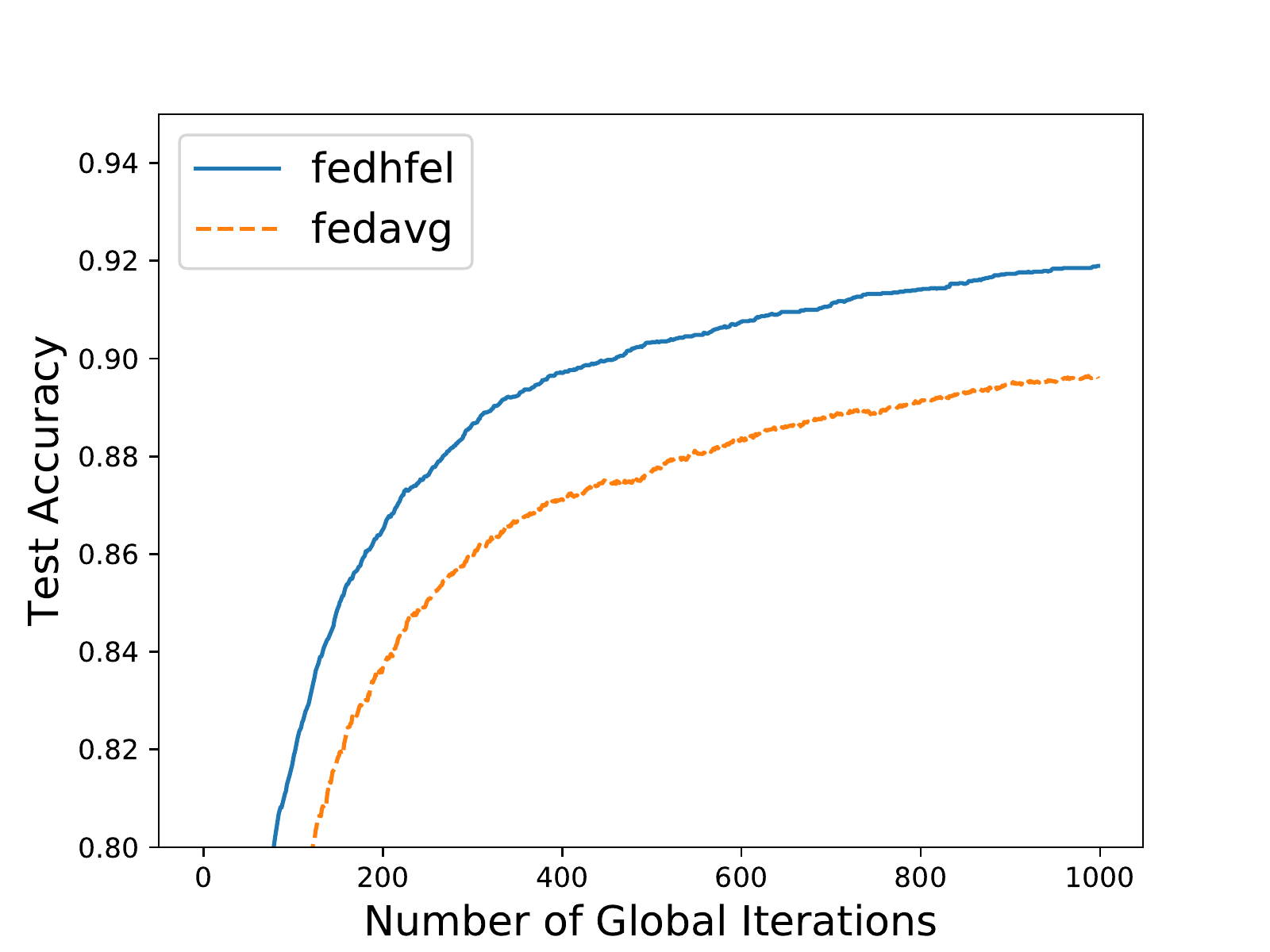}
	\end{minipage}}
	\subfigure[Training accuracy.]{
		\label{train_acc_mnist}
		\begin{minipage}{0.32\textwidth}
			\centering
			\includegraphics[width=\textwidth]{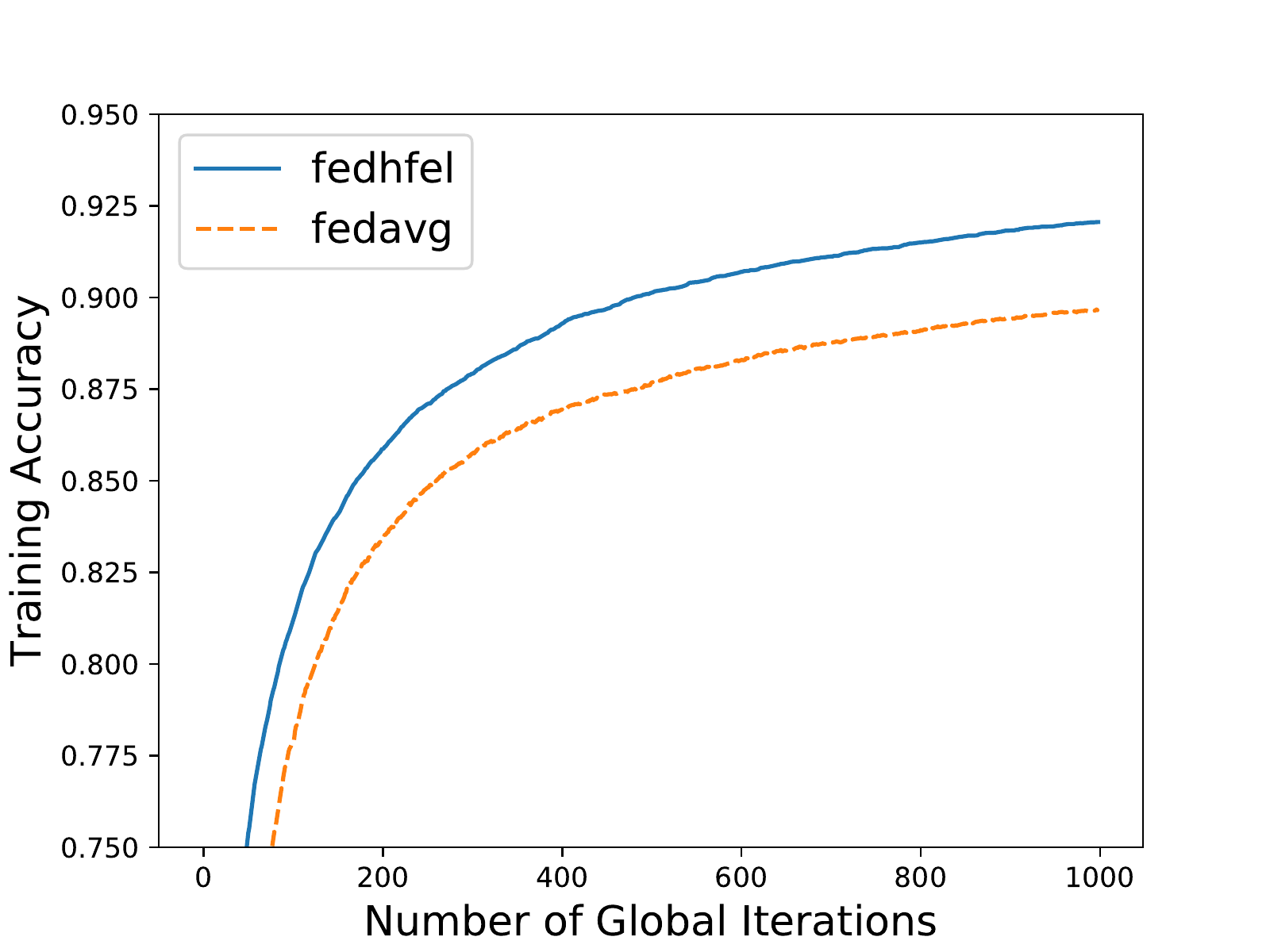}
	\end{minipage}}
	\subfigure[Training loss.]{
		\label{train_loss_mnist}
		\begin{minipage}{0.32\textwidth}
			\centering
			\includegraphics[width=\textwidth]{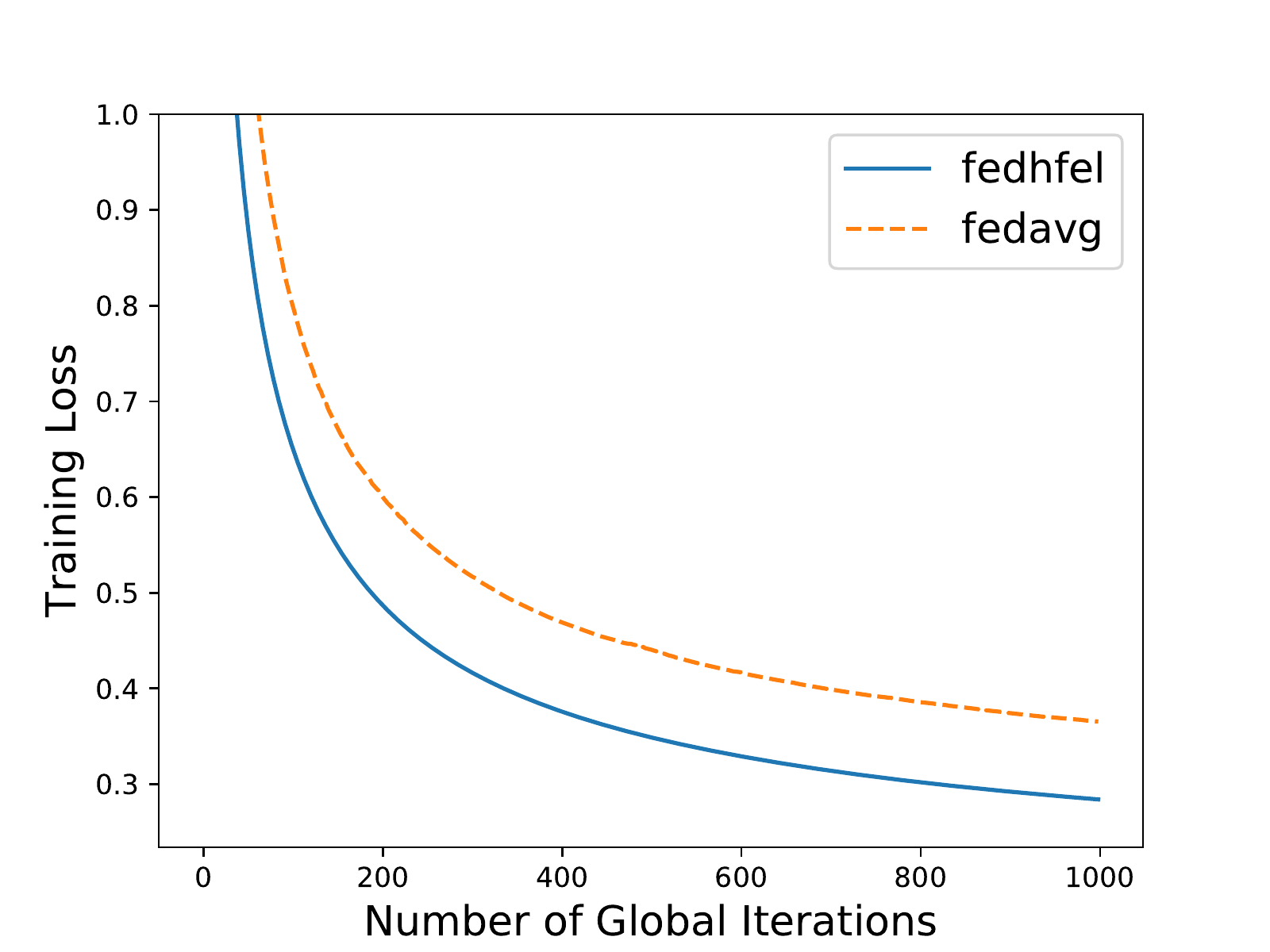}
	\end{minipage}}
	\centering
	\caption{Training results under MNIST.}\label{result_mnist}
\end{figure*}
\begin{figure*}[t]
	\centering
	\subfigure[Test accuracy.]{
		\label{glob_acc_nist}
		\begin{minipage}{0.31\textwidth}
			\centering
			\includegraphics[width=\textwidth]{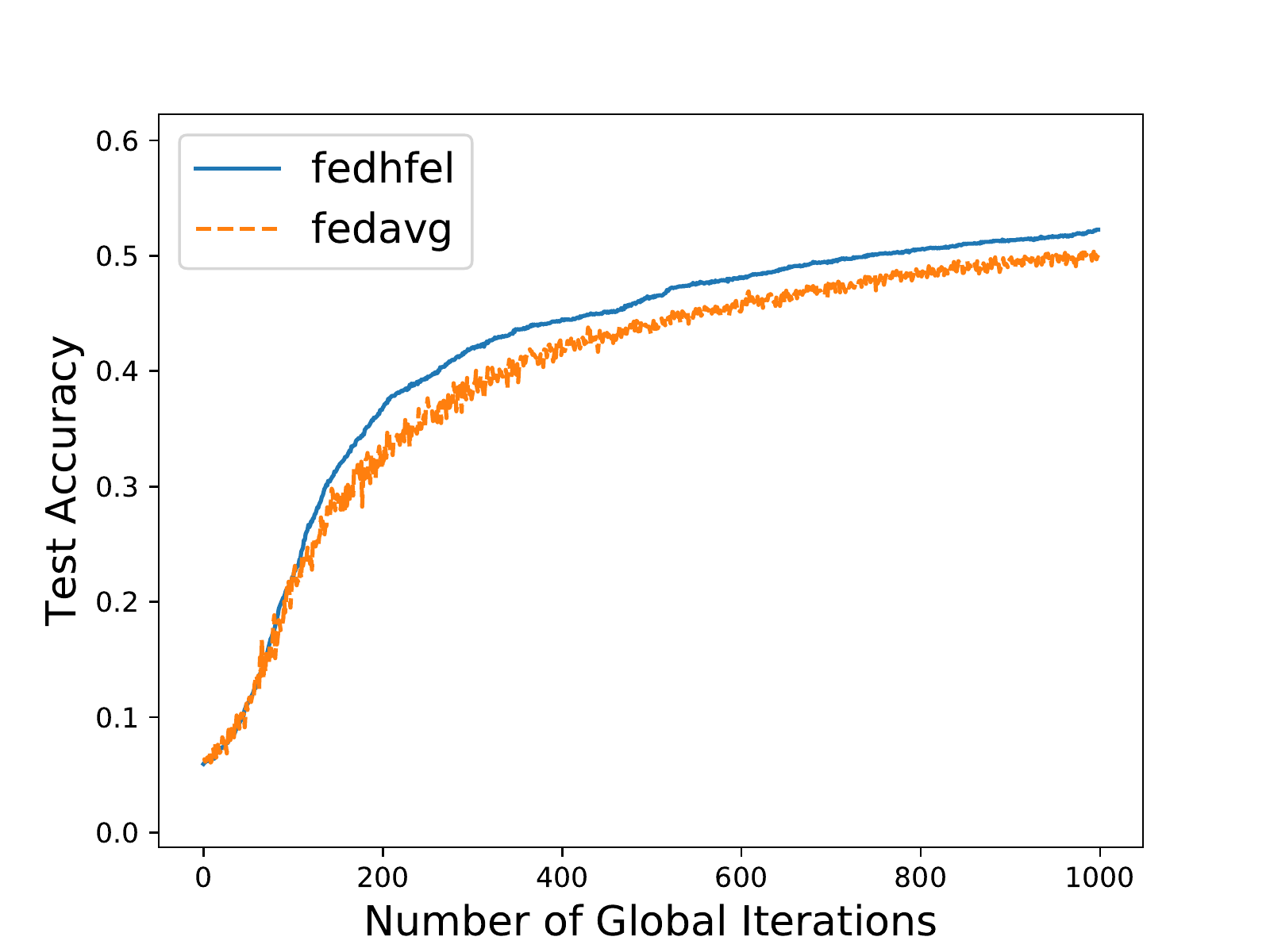}
	\end{minipage}}
	\subfigure[Training accuracy.]{
		\label{train_acc_nist}
		\begin{minipage}{0.32\textwidth}
			\centering
			\includegraphics[width=\textwidth]{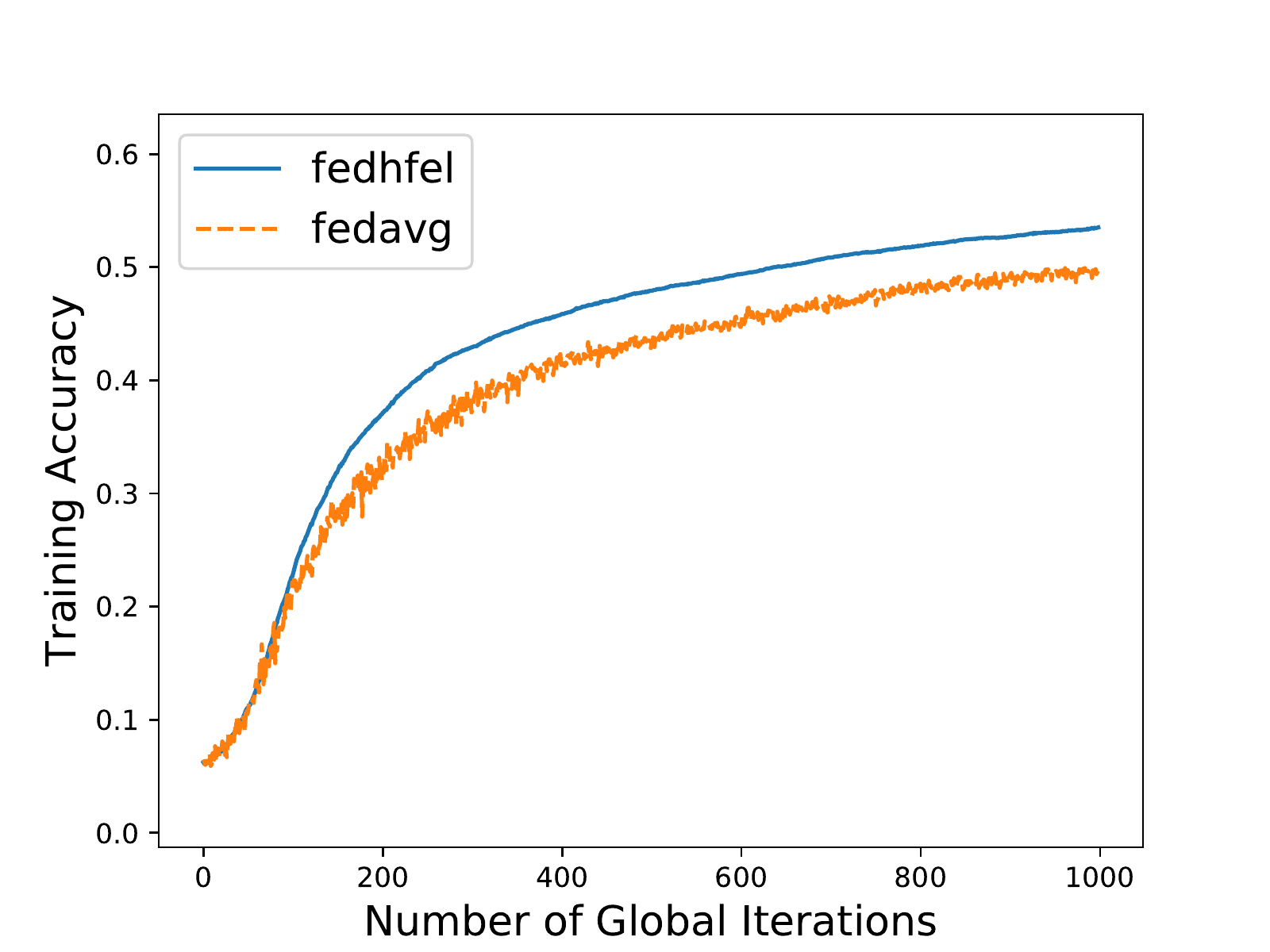}
	\end{minipage}}
	\subfigure[Training loss.]{
		\label{train_loss_nist}
		\begin{minipage}{0.32\textwidth}
			\centering
			\includegraphics[width=\textwidth]{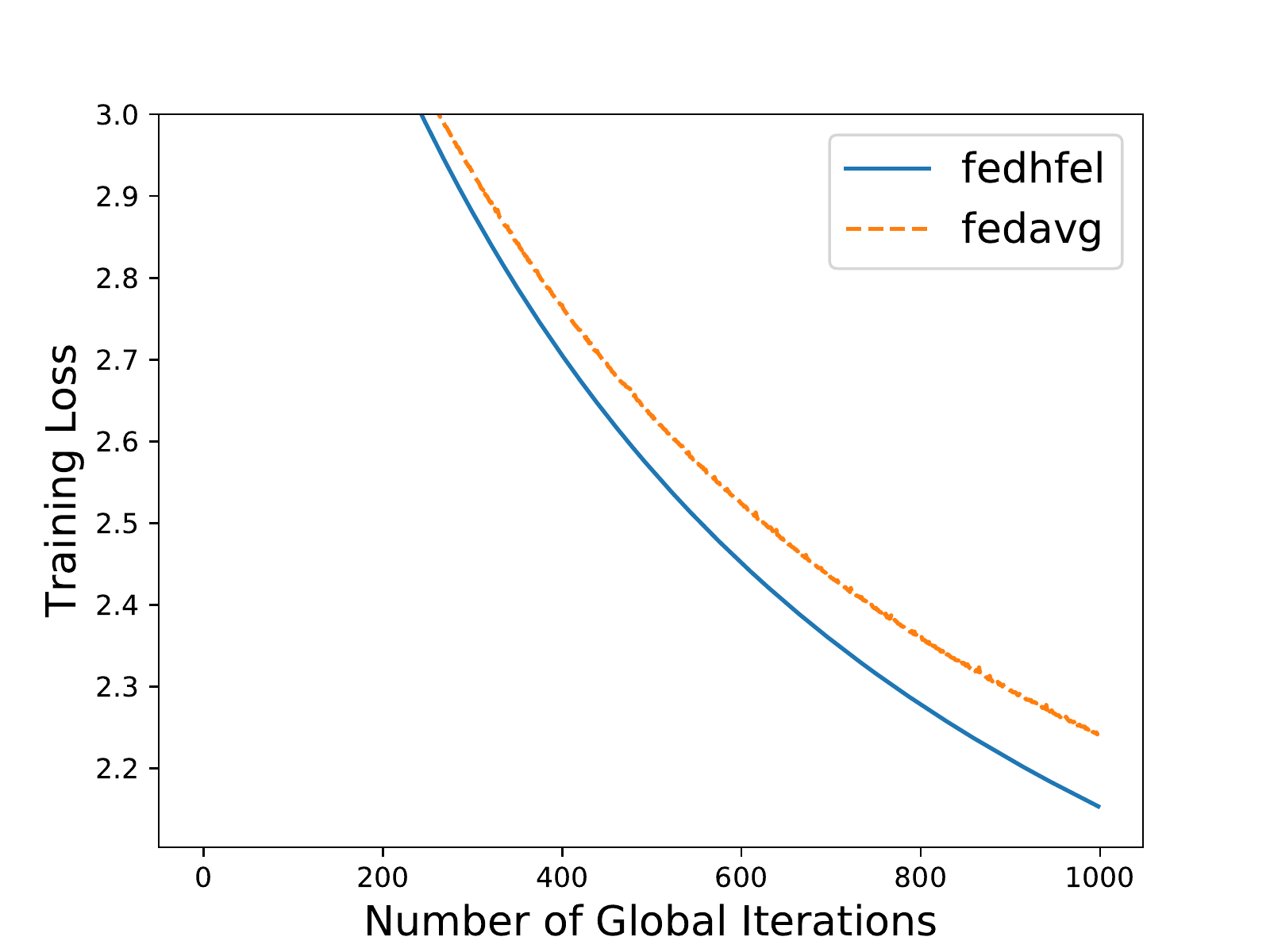}
	\end{minipage}}
	\centering
	\caption{Training results under FEMNIST.}\label{result_nist}
\end{figure*}

Fig. \ref{glob_acc_mnist}-\ref{train_loss_mnist} demonstrate test accuracy, training accuracy and training loss respectively on MNIST dataset as global iteration grows. As is shown, our HFEL algorithm has higher test accuracy and training accuracy than FedAvg both by around $5\%$. And HFEL has lower training loss than FedAvg by around $3\%$. That is for the fact that based on the same number of local iterations during one global iteration, devices in HFEL additionally undergo several rounds of model aggregation in edge servers such that they benefit from model updates at the edge. However for the devices in FedAvg, they only train with local datasets without receiving information from external network for learning improvement during a global iteration. 

Fig. \ref{glob_acc_nist}-\ref{train_loss_nist} present the training performance on dataset FEMNIST. Compared to FedAvg, the increments in terms of test accuracy and training accuracy of HFEL under FEMNIST are up to $4.4\%$ and $4.0\%$ respectively. While the reduction of training loss of HFEL compared with FedAvg under FEMNIST is around $4.1\%$. Because of larger number of data samples in each device and less number of labels to learn for MNIST than FEMNIST, HFEL reveals a higher accuracy and lower training loss on MNIST than FEMNIST dataset in Fig. \ref{result_mnist} and Fig. \ref{result_nist}. Hence it can be assumed that due to the characteristics naturally capturing device heterogeneity, FEMNIST dataset generated by partitioning data based on MNIST would generally obtain a worse learning performance than MNIST. 

\begin{figure}[t]
	\centering
	\begin{minipage}{0.49\textwidth}
		\centering
		\includegraphics[width=\textwidth]{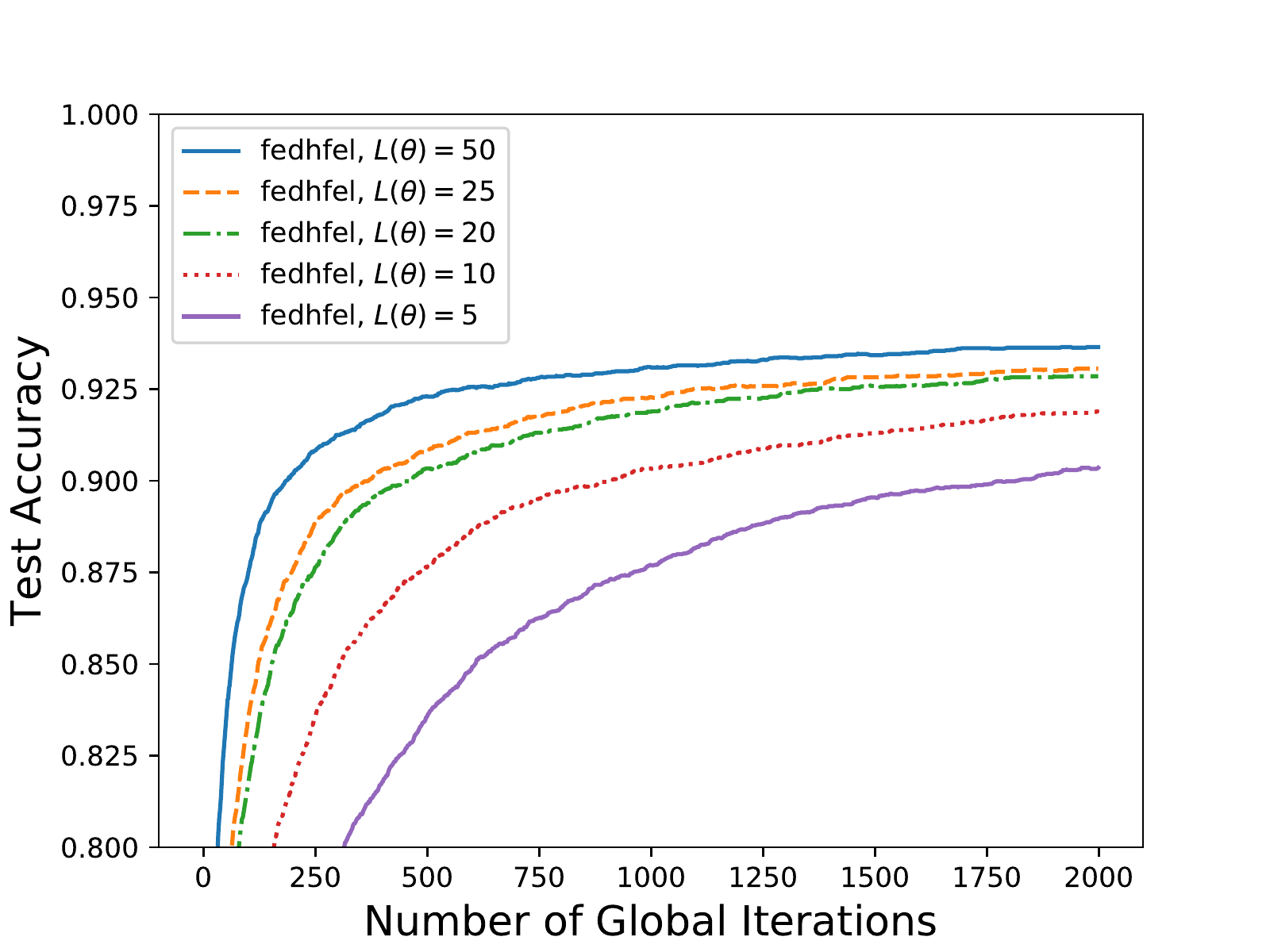}
		\centering
		\caption{Effect of growing local iterations under MNIST.}\label{loc_ep_mnist}
	\end{minipage}
	\hfill
	\begin{minipage}{0.49\textwidth}
		\centering
		\includegraphics[width=\textwidth]{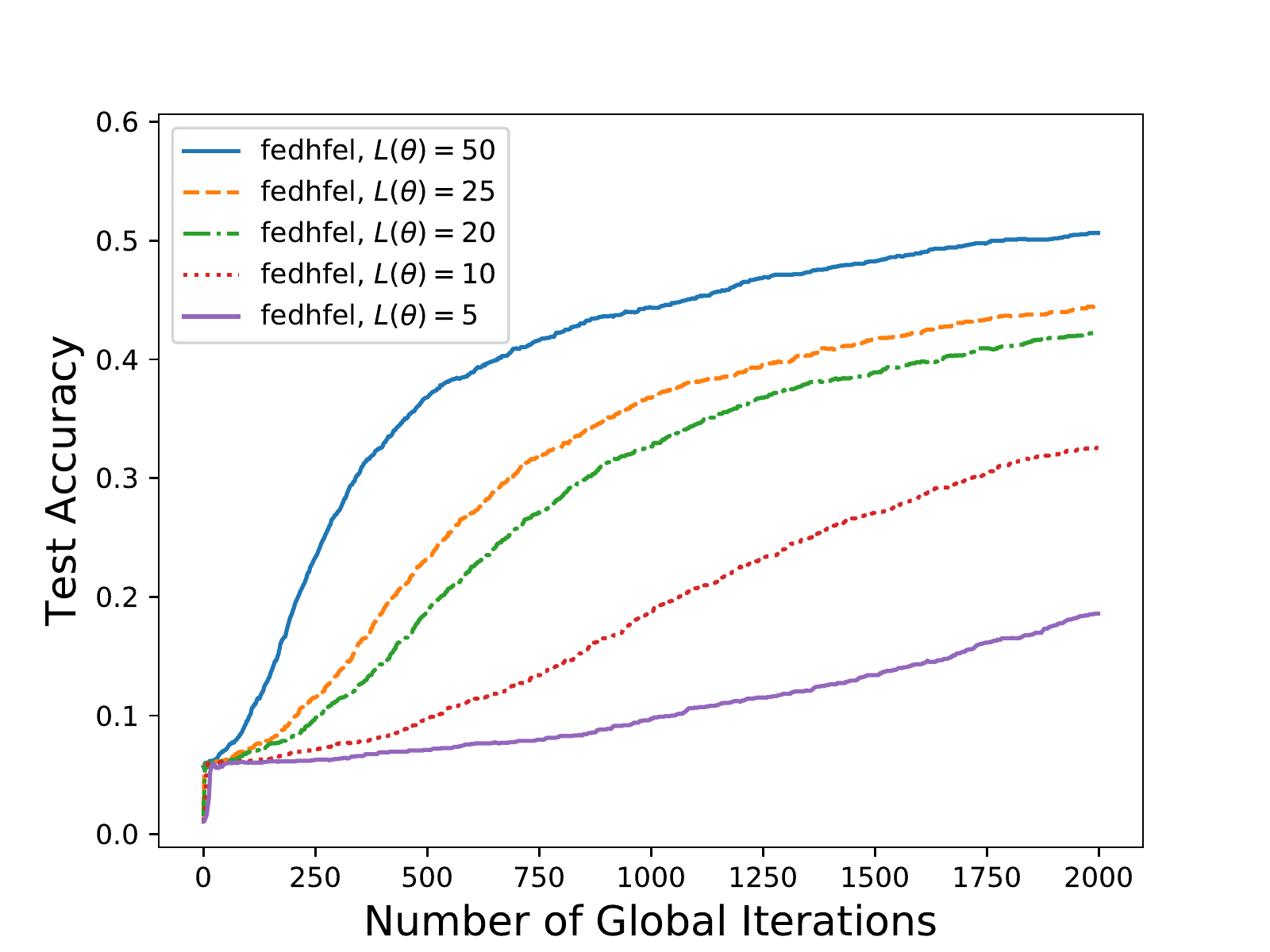}
		\centering
		\caption{Effect of growing local iterations under FEMNIST.}\label{loc_ep_nist}
	\end{minipage}
	\hfill
	\begin{minipage}{0.49\textwidth}
		\centering
		\includegraphics[width=\textwidth]{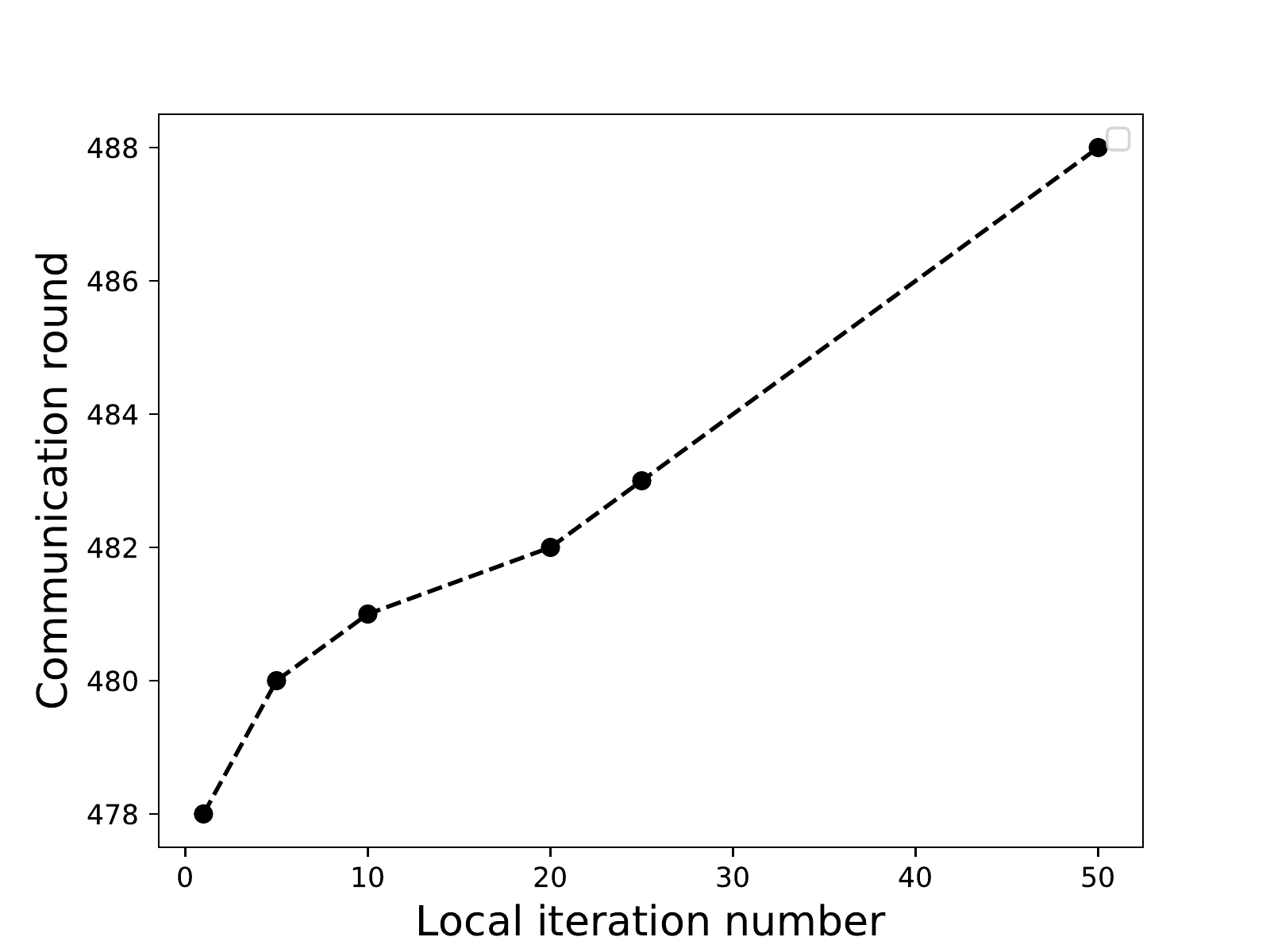}
		\centering
		\caption{Communication rounds with cloud under MNIST.}\label{com_round_mnist}
	\end{minipage}
	\hfill
	\begin{minipage}{0.49\textwidth}
		\centering
		\includegraphics[width=\textwidth]{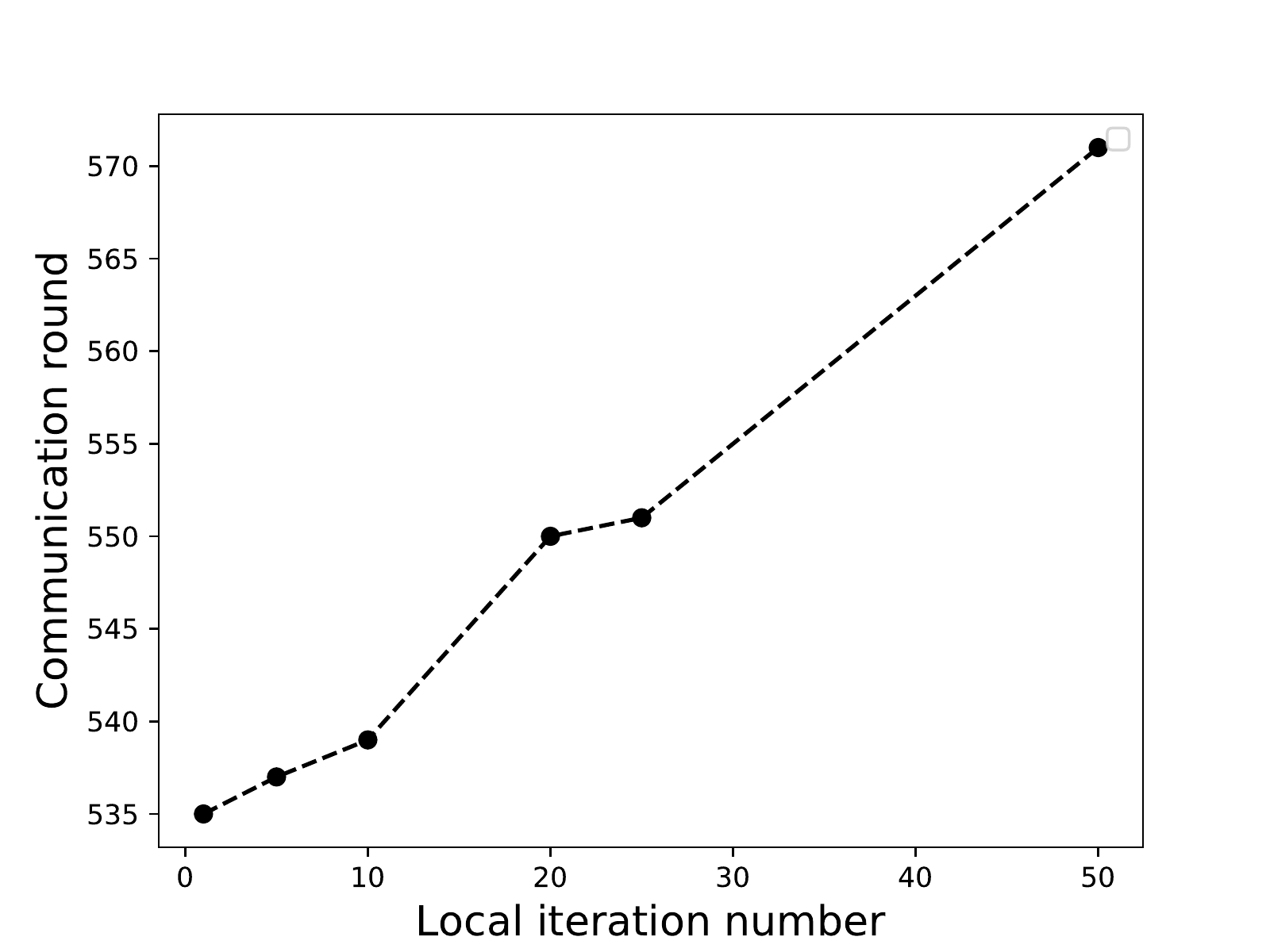}
		\centering
		\caption{Communication rounds with cloud under FEMNIST.}\label{com_round_nist}
	\end{minipage}
\end{figure}

\begin{figure}[t]
	\centering
	\begin{minipage}{0.49\textwidth}
		\centering
		\includegraphics[width=\textwidth]{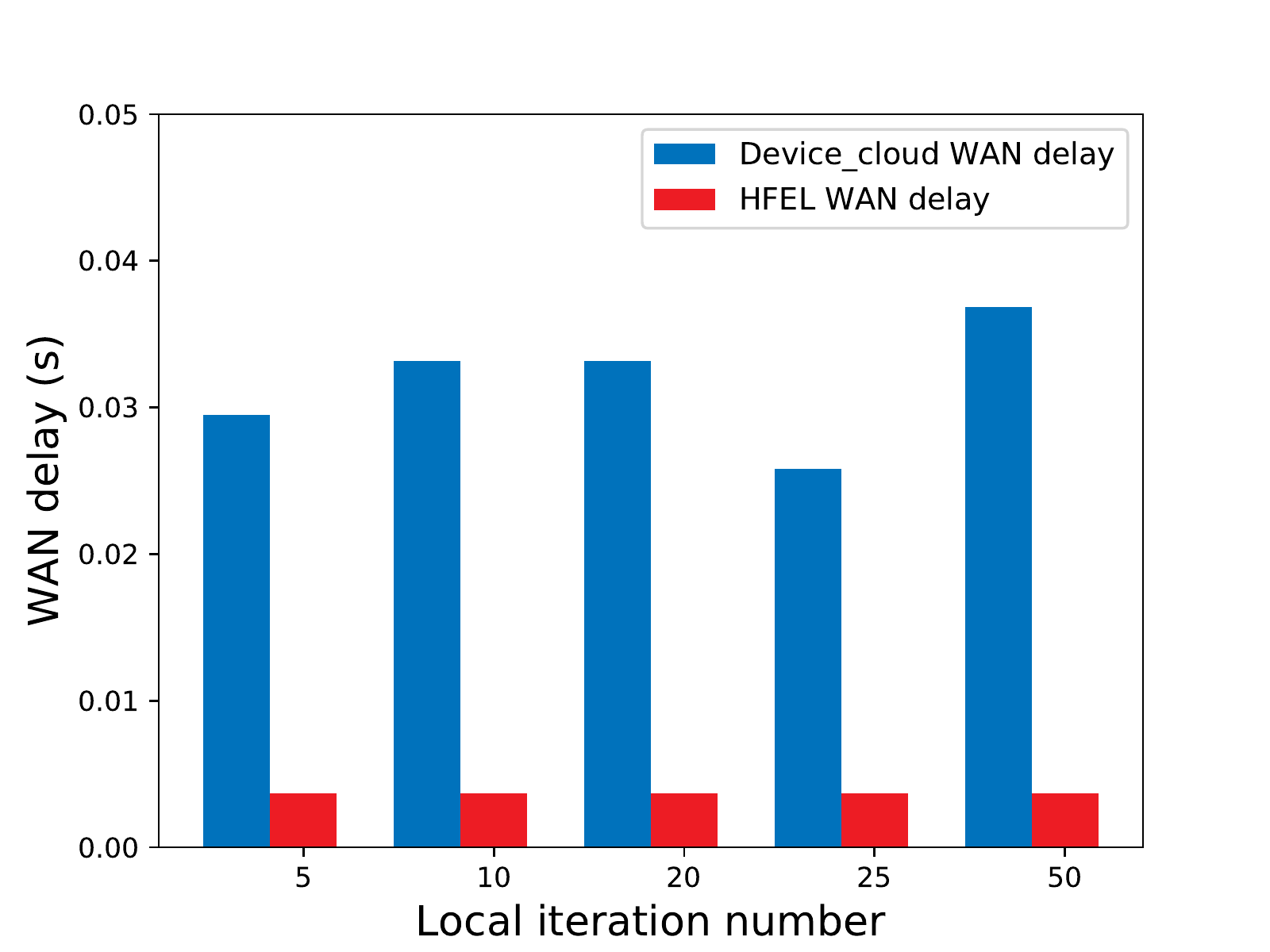}
		\centering
		\caption{WAN communication overhead.}\label{WAN_delay}
	\end{minipage}
	\hfill
	\begin{minipage}{0.49\textwidth}
		\centering
		\includegraphics[width=\textwidth]{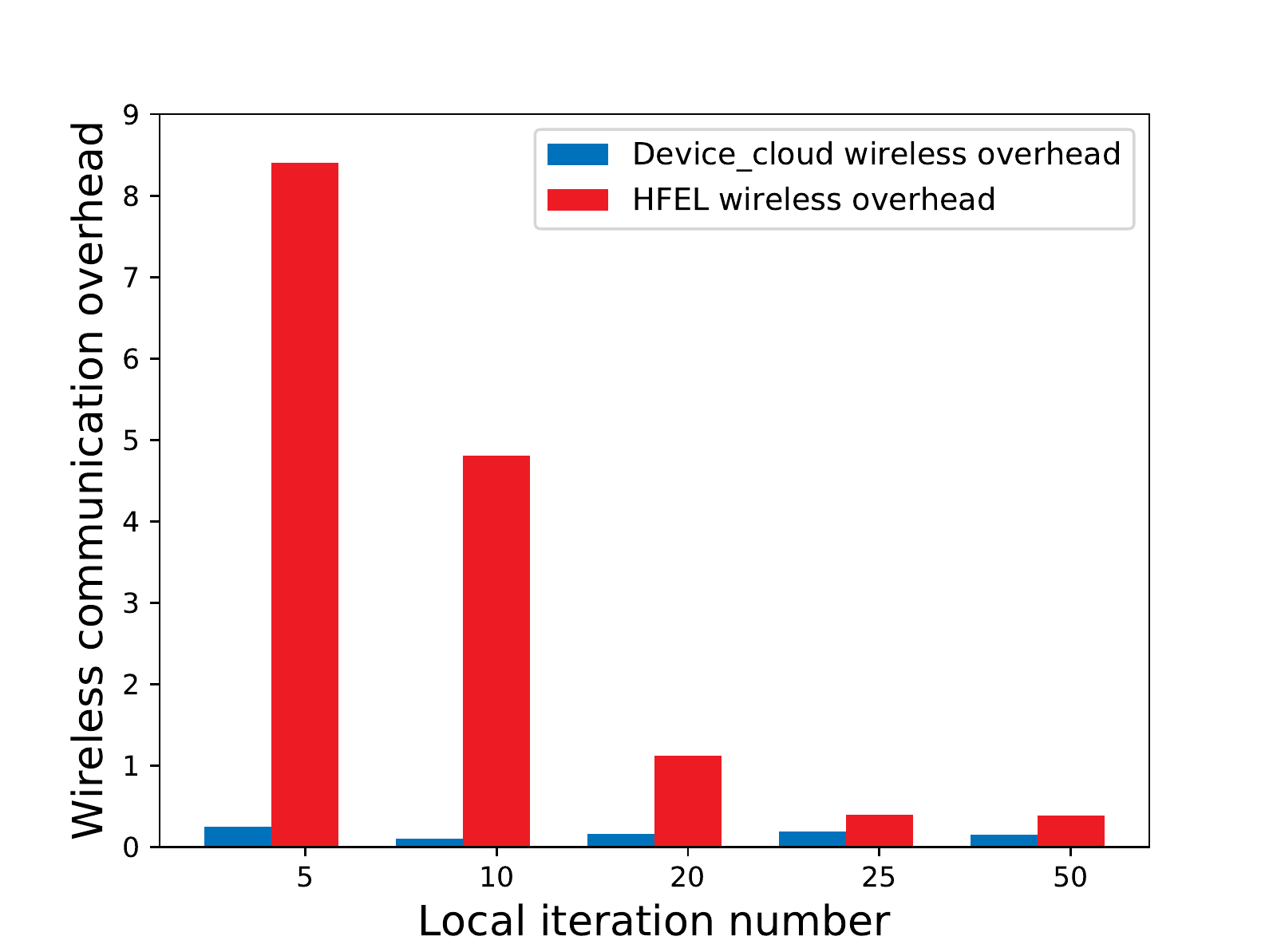}
		\centering
		\caption{Wireless communication overhead.}\label{wireless_overhead}
	\end{minipage}
\end{figure}

The effect of different local iteration numbers $L(\theta)=[5,10,20,25,50]$ on convergence speed is exhibited in Fig. \ref{loc_ep_mnist} and \ref{loc_ep_nist} through $2000$ global iterations. As we can see, with the same number of edge iterations as $5$ and an increase of local iteration number from $5$ to $50$, the convergence speed shows an obvious acceleration both in MNIST and FEMNIST datasets, which implies the growth of $L(\theta)$ has a positive impact on convergence time. 

Then we conduct experiments considering a fixed product of $L(\theta)$ and $I(\epsilon,\theta)$ as $100$ and the values of $L(\theta)$ growing from $1$ to $50$. Fig. \ref{com_round_mnist} and \ref{com_round_nist} show that a decreasing number of local iterations and increasing number of edge iterations lead to a reduction of communication rounds with the cloud to reach the accuracy of $0.9$ for MNIST dataset and $0.55$ for FEMNIST dataset, respectively. Hence, properly increasing edge iteration rounds can help to reduce propagation delay and improve convergence speed in HFEL.

Fig. \ref{WAN_delay} reveals a great advantage of WAN communication efficiency of HFEL over traditional device-cloud FL. Without edge aggregation, there are $N$ devices' local model parameters transmitted through WAN to the remote cloud in device-cloud FL. While in HFEL, after edge aggregation, $K$ (generally $K<<N$) edge servers' edge models, each of which is of similar size to a local model, are transmitted to the cloud. Considerable WAN transmission overheads can be saved in HFEL through edge model aggregation. Fig. \ref{wireless_overhead} shows that wireless communication overhead in HFEL decreases as local iteration number increases. While the wireless overhead of device-cloud FL keeps lower because each device transmits local model to the edge server via wireless connection for only one time. This illustrates that frequent communication between edge servers and devices consumes overhead for wireless data transmission. We should take a careful balance between local iteration number and edge iteration number if our objective turns to minimizing device training overhead.

\section{related work}
\label{related_work}

To date, federated learning (FL) has been envisioned as a promising approach to guarantee personal data security compared to conventional centralized training at the cloud. It only requires local models trained by mobile devices with local datasets to be aggregated by the cloud such that the global model can be updated iteratively until the training process converges.

Nevertheless, faced with long propagation delay in wide-area network (WAN), FL suffers from a bottleneck of communication overhead due to thousands of communication rounds required between mobile devices and the cloud. Hence a majority of studies have focused on reducing communication cost in FL \cite{article2016Jakub,caldas2018expanding,WANG2019CMFL,huang2018loadaboostlossbased}. Authors in \cite{article2016Jakub} proposed structured and sketched local updates to reduce the model size transmitted from mobile devices to the cloud. While authors in \cite{caldas2018expanding} introduced lossy compression and federated dropout to reduce cloud-to-device communication cost, extending the work in \cite{article2016Jakub}. \cite{WANG2019CMFL} figured out a communication-mitigated federated learning (CMFL) algorithm in which devices only upload local updates with high relevance scores to the cloud. Further, considering that communication overhead often dominates computation overhead \cite{McMahan2016McMahan}, authors in \cite{huang2018loadaboostlossbased} increased computation on each device during a local training round by modifying the classic federated averaging algorithm in \cite{McMahan2016McMahan} as LoAdaBoost FedAvg. While in our work, thanks to the emergence of mobile edge computing (MEC) which migrates computing tasks from the network core to the network edge, we propose a hierarchical Federated Edge Learning (HFEL) framework. In HFEL, mobile devices first upload local models to proximate edge servers for partial model aggregation which can offer faster response rate and relieve core network congestion. 

Similarly, some existing literature also proposed hierarchical federated learning in MEC such as \cite{Liu2019EdgeAssisted} which presented a faster convergence speed than the FedAvg algorithm. Although a basic architecture about hierarchical federated learning has been built in \cite{Liu2019EdgeAssisted}, the heterogeneity of mobile device involved in FL is not considered. When large-scale devices with different dataset qualities, computation capacities and battery states participate in FL, resource allocation needs to be optimized to achieve cost efficient training.

There have been several existing research on the resource allocation optimization of mobile devices for different efficiency maximization objectives in edge-assisted FL \cite{Federated2018Yu,nishio2018client,yoshida2019hybridfl,ren2019accelerating,zeng2019energyefficient,dinh2019federated,vu2019cellfree}. Yu et al. worked on federated learning based proactive content caching (FPCC) \cite{Federated2018Yu}. While Nishio et al. proposed an FL protocal called FedCS to maximize the participating number of devices with a predefined deadline based on their wireless channel states and computing capacities \cite{nishio2018client}. Further, the authors extended their study of FedCS to \cite{yoshida2019hybridfl} in which data distribution differences are considered and solved by constructing independent identically distributed (IID) dataset. In \cite{ren2019accelerating}, the authors aimed at accelerating training process via optimizing batchsize selection and communication resource allocation in a federated edge learning (FEEL) framework. \cite{zeng2019energyefficient} explored energy-efficient radio resource management in FL and proposed energy-efficient strategies for bandwidth allocation and edge association. Dinh et al. worked on a resource allocation problem that captures the trade-off between convergence time and energy cost in FL \cite{dinh2019federated}. While in \cite{vu2019cellfree}, local accuracy, transmit power, data rate and devices' computing capacities were jointly optimized for FL training time minimization.

In our HFEL framework, we target at solving computation and bandwidth resource allocation of each device for training cost minimization in terms of energy and delay. Furthermore, edge association is optimized for each edge server under the scenario where more than one edge server is involved in HFEL and each device is able to communicate with multiple edge servers. While the literature \cite{ren2019accelerating,zeng2019energyefficient,dinh2019federated,vu2019cellfree} take only one edge server into account for resource allocation. Along a different line, we work on training cost minimization in terms of energy and delay by considering 1) joint computation and bandwidth resource allocation for each device and 2) edge association for each edge server.

\section{conclusion}
\label{conclusion}

Federated Learning (FL) has been proposed as an appealing approach to handle data security issue of mobile devices compared to conventional machine learning at the remote cloud with raw data. To enable great potentials in low-latency and energy-efficient FL, we introduce hierarchical Federated Edge Learning (HFEL) framework in which model aggregation is partially migrated to edge servers from the cloud. Furthermore, a joint computation and communication resource scheduling model under HFEL framework is formulated to achieve global cost minimization. Yet proving the minimization problem owns extremely high time complexity, we devise an efficient resource scheduling algorithm which can be decomposed into two subproblems: resource allocation given a scheduled set of devices for each edge server and edge association for all the edge servers. Through cost reducing iterations of solving resource allocation and edge association, our proposed HFEL algorithm terminates to a stable system point where it fulfills substantial performance gain in cost reduction compared with the benchmarks.

Eventually, compared to conventional federated learning without edge servers as intermediaries \cite{McMahan2016McMahan}, the HFEL framework accomplishes higher global and test accuracies and lower training loss as our simulation results show.

\bibliographystyle{IEEEtran}
\bibliography{ref}

\begin{thebibliography}{10}
\providecommand{\url}[1]{#1}
\csname url@samestyle\endcsname
\providecommand{\newblock}{\relax}
\providecommand{\bibinfo}[2]{#2}
\providecommand{\BIBentrySTDinterwordspacing}{\spaceskip=0pt\relax}
\providecommand{\BIBentryALTinterwordstretchfactor}{4}
\providecommand{\BIBentryALTinterwordspacing}{\spaceskip=\fontdimen2\font plus
\BIBentryALTinterwordstretchfactor\fontdimen3\font minus
  \fontdimen4\font\relax}
\providecommand{\BIBforeignlanguage}[2]{{%
\expandafter\ifx\csname l@#1\endcsname\relax
\typeout{** WARNING: IEEEtran.bst: No hyphenation pattern has been}%
\typeout{** loaded for the language `#1'. Using the pattern for}%
\typeout{** the default language instead.}%
\else
\language=\csname l@#1\endcsname
\fi
#2}}
\providecommand{\BIBdecl}{\relax}
\BIBdecl

\bibitem{Zhou2019Edge}
Z.~{Zhou}, X.~{Chen}, E.~{Li}, L.~{Zeng}, K.~{Luo}, and J.~{Zhang}, ``Edge
  intelligence: Paving the last mile of artificial intelligence with edge
  computing,'' \emph{Proceedings of the IEEE}, vol. 107, no.~8, pp. 1738--1762,
  Aug 2019.

\bibitem{Goodfellow-et-al-2016}
I.~Goodfellow, Y.~Bengio, and A.~Courville, \emph{Deep learning}.\hskip 1em
  plus 0.5em minus 0.4em\relax MIT press, 2016,
  \url{http://www.deeplearningbook.org}.

\bibitem{LI201776}
P.~Li, J.~Li, Z.~Huang, T.~Li, C.-Z. Gao, S.-M. Yiu, and K.~Chen, ``Multi-key
  privacy-preserving deep learning in cloud computing,'' \emph{Future
  Generation Computer Systems}, vol.~74, pp. 76 -- 85, 2017.

\bibitem{book2019Bart}
B.~Custers, A.~M. Sears, F.~Dechesne, I.~Georgieva, T.~Tani, and S.~van~der
  Hof, \emph{EU Personal Data Protection in Policy and Practice}.\hskip 1em
  plus 0.5em minus 0.4em\relax Springer, 2019.

\bibitem{Privacy2014Gaff}
B.~M. {Gaff}, H.~E. {Sussman}, and J.~{Geetter}, ``Privacy and big data,''
  \emph{Computer}, vol.~47, no.~6, pp. 7--9, June 2014.

\bibitem{Consumer2013}
\BIBentryALTinterwordspacing
A.~Anonymous, ``Consumer data privacy in a networked world: A framework for
  protecting privacy and promoting innovation in the global digital economy,''
  \emph{Journal of Privacy and Confidentiality}, vol.~4, no.~2, Mar. 2013.
  [Online]. Available:
  \url{https://journalprivacyconfidentiality.org/index.php/jpc/article/view/623}
\BIBentrySTDinterwordspacing

\bibitem{konen2016federated}
J.~Kone{\v{c}}n{\`y}, H.~B. McMahan, D.~Ramage, and P.~Richt{\'a}rik,
  ``Federated optimization: Distributed machine learning for on-device
  intelligence,'' \emph{arXiv preprint arXiv:1610.02527}, 2016.

\bibitem{McMahan2016McMahan}
H.~{Brendan McMahan}, E.~{Moore}, D.~{Ramage}, S.~{Hampson}, and B.~{Ag{\"u}era
  y Arcas}, ``{Communication-Efficient Learning of Deep Networks from
  Decentralized Data},'' \emph{Artificial Intelligence and Statistics}, pp.
  1273--1282, Apr 2017.

\bibitem{article2004Maxim}
M.~Sviridenko, ``A note on maximizing a submodular set function subject to a
  knapsack constraint,'' \emph{Operations Research Letters}, vol.~32, no.~1,
  pp. 41--43, 2004.

\bibitem{Edge2016Shi}
W.~{Shi}, J.~{Cao}, Q.~{Zhang}, Y.~{Li}, and L.~{Xu}, ``Edge computing: Vision
  and challenges,'' \emph{IEEE Internet of Things Journal}, vol.~3, no.~5, pp.
  637--646, Oct 2016.

\bibitem{chen2015efficient}
X.~Chen, L.~Jiao, W.~Li, and X.~Fu, ``Efficient multi-user computation
  offloading for mobile-edge cloud computing,'' \emph{IEEE/ACM Transactions on
  Networking}, vol.~24, no.~5, pp. 2795--2808, 2015.

\bibitem{li2019edge}
E.~{Li}, L.~{Zeng}, Z.~{Zhou}, and X.~{Chen}, ``Edge ai: On-demand accelerating
  deep neural network inference via edge computing,'' \emph{IEEE Transactions
  on Wireless Communications}, vol.~19, no.~1, pp. 447--457, Jan 2020.

\bibitem{EnergyEfficient2017You}
C.~{You}, K.~{Huang}, H.~{Chae}, and B.~{Kim}, ``Energy-efficient resource
  allocation for mobile-edge computation offloading,'' \emph{IEEE Transactions
  on Wireless Communications}, vol.~16, no.~3, pp. 1397--1411, March 2017.

\bibitem{Konecny2014Semi}
J.~Kone{\v{c}}n{\`y}, Z.~Qu, and P.~Richt{\'a}rik, ``Semi-stochastic coordinate
  descent,'' \emph{Optimization Methods and Software}, vol.~32, no.~5, pp.
  993--1005, 2017.

\bibitem{dinh2019federated}
C.~Dinh, N.~H. Tran, M.~N. Nguyen, C.~S. Hong, W.~Bao, A.~Zomaya, and
  V.~Gramoli, ``Federated learning over wireless networks: Convergence analysis
  and resource allocation,'' \emph{arXiv preprint arXiv:1910.13067}, 2019.

\bibitem{Burd1996Processor}
T.~D. Burd and R.~W. Brodersen, ``Processor design for portable systems,''
  \emph{Journal of Vlsi Signal Processing Systems for Signal Image \& Video
  Technology}, vol.~13, no. 2-3, pp. 203--221, 1996.

\bibitem{vu2019cellfree}
T.~T. Vu, D.~T. Ngo, N.~H. Tran, H.~Q. Ngo, M.~N. Dao, and R.~H. Middleton,
  ``Cell-free massive mimo for wireless federated learning,'' \emph{arXiv
  preprint arXiv:1909.12567}, 2019.

\bibitem{Ma2015Distributed}
C.~Ma, J.~Kone{\v{c}}n{\`y}, M.~Jaggi, V.~Smith, M.~I. Jordan,
  P.~Richt{\'a}rik, and M.~Tak{\'a}{\v{c}}, ``Distributed optimization with
  arbitrary local solvers,'' \emph{Optimization Methods and Software}, vol.~32,
  no.~4, pp. 813--848, 2017.

\bibitem{TP-toolbox-web}
``Stirling number of the second kind,''
  \url{http://mathworld.wolfram.com/StirlingNumberoftheSecondKind.html}.

\bibitem{Lecun1998GradientBased}
Y.~Lecun, L.~Bottou, Y.~Bengio, and P.~Haffner, ``Gradient-based learning
  applied to document recognition,'' \emph{Proceedings of the IEEE}, vol.~86,
  pp. 2278 -- 2324, 12 1998.

\bibitem{li2018federated}
T.~Li, A.~K. Sahu, M.~Zaheer, M.~Sanjabi, A.~Talwalkar, and V.~Smith,
  ``Federated optimization in heterogeneous networks,'' \emph{arXiv preprint
  arXiv:1812.06127}, 2018.

\bibitem{caldas2018leaf}
S.~Caldas, P.~Wu, T.~Li, J.~Kone{\v{c}}n{\`y}, H.~B. McMahan, V.~Smith, and
  A.~Talwalkar, ``Leaf: A benchmark for federated settings,'' \emph{arXiv
  preprint arXiv:1812.01097}, 2018.

\bibitem{article2016Jakub}
J.~Kone{\v{c}}n{\`y}, H.~B. McMahan, F.~X. Yu, P.~Richt{\'a}rik, A.~T. Suresh,
  and D.~Bacon, ``Federated learning: Strategies for improving communication
  efficiency,'' \emph{arXiv preprint arXiv:1610.05492}, 2016.

\bibitem{caldas2018expanding}
S.~Caldas, J.~Kone{\v{c}}ny, H.~B. McMahan, and A.~Talwalkar, ``Expanding the
  reach of federated learning by reducing client resource requirements,''
  \emph{arXiv preprint arXiv:1812.07210}, 2018.

\bibitem{WANG2019CMFL}
L.~{WANG}, W.~{WANG}, and B.~{LI}, ``Cmfl: Mitigating communication overhead
  for federated learning,'' in \emph{2019 IEEE 39th International Conference on
  Distributed Computing Systems (ICDCS)}, July 2019, pp. 954--964.

\bibitem{huang2018loadaboostlossbased}
L.~Huang, Y.~Yin, Z.~Fu, S.~Zhang, H.~Deng, and D.~Liu, ``Loadaboost:
  Loss-based adaboost federated machine learning on medical data,'' \emph{arXiv
  preprint arXiv:1811.12629}, 2018.

\bibitem{Liu2019EdgeAssisted}
L.~Liu, J.~Zhang, S.~Song, and K.~B. Letaief, ``Edge-assisted hierarchical
  federated learning with non-iid data,'' \emph{arXiv preprint
  arXiv:1905.06641}, 2019.

\bibitem{Federated2018Yu}
Z.~{Yu}, J.~{Hu}, G.~{Min}, H.~{Lu}, Z.~{Zhao}, H.~{Wang}, and N.~{Georgalas},
  ``Federated learning based proactive content caching in edge computing,'' in
  \emph{2018 IEEE Global Communications Conference (GLOBECOM)}, Dec 2018, pp.
  1--6.

\bibitem{nishio2018client}
T.~{Nishio} and R.~{Yonetani}, ``Client selection for federated learning with
  heterogeneous resources in mobile edge,'' in \emph{ICC 2019 - 2019 IEEE
  International Conference on Communications (ICC)}, May 2019, pp. 1--7.

\bibitem{yoshida2019hybridfl}
N.~Yoshida, T.~Nishio, M.~Morikura, K.~Yamamoto, and R.~Yonetani, ``Hybrid-fl
  for wireless networks: Cooperative learning mechanism using non-iid data,''
  \emph{arXiv preprint arXiv:1905.07210}, 2019.

\bibitem{ren2019accelerating}
J.~Ren, G.~Yu, and G.~Ding, ``Accelerating dnn training in wireless federated
  edge learning system,'' \emph{arXiv preprint arXiv:1905.09712}, 2019.

\bibitem{zeng2019energyefficient}
Q.~Zeng, Y.~Du, K.~K. Leung, and K.~Huang, ``Energy-efficient radio resource
  allocation for federated edge learning,'' \emph{arXiv preprint
  arXiv:1907.06040}, 2019.

\end{thebibliography}




%

\end{document}